\documentclass[final]{lipics-v2021}

\nolinenumbers
\hideLIPIcs
\usepackage{wrapfig}
\usepackage{graphicx} 
\usepackage{amssymb,amsmath}
\usepackage{algorithm}
\usepackage[noend]{algpseudocode}
\usepackage{comment}
\usepackage{tikz}
\usepackage{caption}
\usepackage{subcaption}
\usetikzlibrary{arrows.meta, decorations.pathmorphing}

\newif\iffull\fullfalse

\newcommand{\E}{{\cal E}}

\newcommand{\PP}{{\cal P}}
\newcommand{\SP}{{\cal SP}}
\newcommand{\e}{{\varepsilon}}
\newcommand{\x}{{\xi}}
\newcommand{\C}{{\mathcal{C}}}

\renewcommand{\emptyset}{\varnothing}

\newcommand{\NP}{\emph{NP}}
\newcommand{\edge}{\!\rightarrow\!}
\usepackage{tikz-cd}
\usetikzlibrary{decorations.pathmorphing}
\usepackage[customcolors]{hf-tikz}
\hfsetfillcolor{gray!20} 
\hfsetbordercolor{white}
\newcommand{\pat}{\rightsquigarrow}
\newcommand{\W}{\cal W}

\title{A Decremental Algorithm for Checking the Possibility of Braess Paradox in Dynamic Nets}
\author{Dario {Fiorenza}}{Sapienza, University of Rome, Italy}{fiorenza@di.uniroma1.it}{}{}
\author{Daniele {Gorla}}{Sapienza, University of Rome, Italy}{gorla@di.uniroma1.it}{}{}
\author{Ivano {Salvo}}{Sapienza, University of Rome, Italy}{salvo@di.uniroma1.it}{}{}
\authorrunning{D. Fiorenza, D. Gorla, I. Salvo}
\Copyright{ }


\date{October 2026}

\begin{document}

\maketitle

\begin{abstract}
Braess paradox originates when latency at Wardrop equilibrium in traffic networks decreases because of removing edges. The
graph-theoretic property of networks suffering from the Braess paradox
was called vulnerability by Roughgarden in 2006; it was then characterized
and algorithmically checked both for undirected and for directed nets. 
In this paper, we provide a decremental algorithm of linear amortized
complexity to
check vulnerability for dynamically evolving networks. The basic idea of our dynamic algorithm is to use a static
algorithm that marks some edges as irrelevant for the vulnerability of
the graph and ignore those edges for all subsequent runs of the decremental
procedure. To get a linear amortized cost for every edge remotion, we
also provide a new version of such a static algorithm that improves its
complexity from $\mathcal{O}(nm^2)$ to $\mathcal{O}(m^2)$, that is in turn aligned with the cost
of the best state-of-the-art static algorithm for vulnerability.

\end{abstract}

\section{Introduction}
\vspace{-2mm}
In selfish routing, non-cooperative agents travel from a source node $s$ to a destination node $t$, choosing their path in order to minimize their own  latency. 
{\em Traffic networks} \cite{braessFormal,BI97} provide a model for studying selfish routing:  
edges of a directed graph are labeled by a function 
specifying the latency that agents experience traveling on an edge 
in terms of the flow passing through it. 
Since the latency experienced by an agent 
while traveling along a path depends on network congestion, 
traffic in a network is modeled as a non-cooperative game, 
which stabilizes to equilibrium when all agents cannot improve their latency by 
choosing a different path.
This notion of equilibrium was defined by Wardrop \cite{War52} in the
context of transport analysis.

As a representative example, consider the so-called {\em Wheatstone network} in Fig.~\ref{fig:Wheat}(1).
The edges $u\edge v$, $s\edge v$ and $u\edge t$ are  labeled by the constant functions 0 and 1: 
independently on how much traffic travels along them,
the passage from $u$ to $v$ is instantaneous, whereas it creates a delay of 1 in the other two cases. 
Finally, the delay on edges $s\edge u$ and $v\edge t$ 
is linear in the amount of traffic traveling along them.
The (non-optimal) flow of value $1$ where all selfish agents choose the path $s\edge u\edge v\edge t$ is at Wardrop equilibrium, 
and the overall delay experienced is $1+0+1=2$.
In contrast, the (optimal) flow in which agents distribute half along the path $s\edge u\edge t$ and half 
along $s\edge v\edge t$ is {\em not} at equilibrium, but the delay experienced is smaller (viz., 3/2). 
This is indeed the delay obtained with a flow of 
value 1 at Wardrop equilibrium in the subnet of the Wheatstone network depicted in Fig.~\ref{fig:Wheat}(2).
{\em Braess paradox} \cite{braessFormal,braessOriginal} 
originates when latency at Wardrop equilibrium decreases because of removing edges:
the Wheatstone network is a minimal example of this counterintuitive phenomenon. 

\begin{figure}[t]
		\vspace{-2mm}
	\begin{tabular}{ccc}
		\begin{minipage}{0.29\textwidth}
			\center
			\begin{tikzcd}[column sep=.8cm,row sep=.4cm]
			& u \ar[rd, "1"]\ar[dd,"0"]\\
			s \ar[ur,"x"] \ar[dr,"1" '] && t\\
			& v \ar[ru,"x" ']\\
			&(1)
			\\
			\end{tikzcd}
		\end{minipage}
		&
		\begin{minipage}{0.29\textwidth}
			\center
			\begin{tikzcd}[column sep=.8cm,row sep=.4cm]
			& u \ar[rd,"1"]\\
			s \ar[ur,"x"] \ar[dr,"1" '] && t\\
			& v \ar[ru,"x" ']\\
			&(2)
			\\
			\end{tikzcd}
		\end{minipage}
		&
		\begin{minipage}{0.33\textwidth}
			\center
			\begin{tikzcd}[column sep=.8cm,row sep=.4cm]
			& u \ar[dd]\ar[dr] & \\
			s \ar[ru]\ar[rd] & & t\\ 
			& v \ar[ru]\\
			&(3)
			\\
        \vspace{-8mm}    
			\end{tikzcd}
		\end{minipage}
		\vspace{-6mm}
	\end{tabular}
	\caption{(1) The Wheatstone network; (2) Its optimal subnet; (3) The graph $\W$.}
	\label{fig:Wheat}
	\vspace{-5mm}
\end{figure}
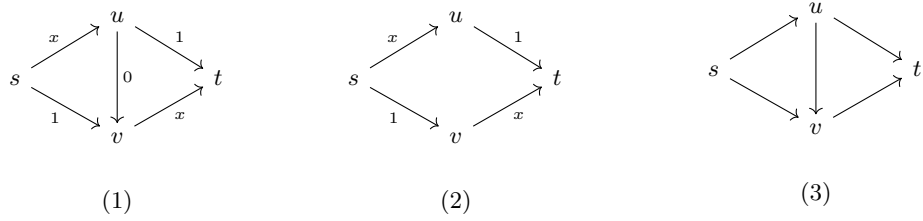 

Braess paradox has been studied for decades and from many different perspectives. An important result is that the paradox is not merely an artificially constructed phenomenon arising under particular conditions; rather, in random networks, the probability of its occurrence tends to one as the size of the network increases \cite{VR06}. The paradox has therefore also been investigated in real-world applications, such as gas transportation networks \cite{ALB12}, electric power systems \cite{BS06}, and road networks \cite{BSAMBC13}.
From a theoretical perspective, several results guarantee the absence of the paradox, either through conditions on the flow distribution \cite{HA01} or on the structure of the strategies \cite{Matroids17}. Moreover, 
it has been investigated 
how detrimental the paradox can be in different settings, as well as how much the addition of the bridge could actually improve equilibria in some cases \cite{Penchina97}. Furthermore, it has been proved that the paradox can occur only for a bounded set of traffic demand values \cite{PP97} and some conditions under which the paradox arises have been devised when cooperative and competitive agents coexist \cite{MixedEquilibrium16}.

From an algorithmic point of view, checking if a network suffers from the Braess paradox 
has been proven to be a \NP-hard problem in \cite{Rough06}, where a change of perspective has been suggested: studying the Braess paradox from a graph-theoretical point of view, by characterizing and efficiently recognizing {\em vulnerable} graphs, i.e. those graphs that admit instances (i.e., assignments of latency functions to edges and traffic demand) that generate the paradox. 
A characterization of vulnerable undirected multigraphs is presented in~\cite{Milch06}, where
it is proved that an undirected graph is vulnerable if and only if it is not series-parallel \cite{RS42}. This characterization only holds for graphs where all nodes and edges lie on at least one 
simple $st$-path. In \cite{ChenEtal15}, the same characterization is proved for directed multigraphs that satisfy the very same condition, called {\em irredundancy} therein.  While checking (ir)redundancy 
and finding the \emph{Maximal Irredundant Subgraph} (MIS) 
can be efficiently done for undirected graphs \cite{CinesiFull}, the same does not hold for directed ones: in \cite{CGS19} it is proved that recognizing if an edge is redundant in a digraph is a \NP-hard problem.
Moreover, in \cite{CGS18} the authors prove that vulnerability coincides with containing an irredundant subgraph homeomorphic to the graph $\W$ of Fig.~\ref{fig:Wheat}(3)
and in \cite{CGS19} they provide a $\mathcal{O}(nm^2)$ algorithm to state if a given graph is vulnerable or not; in \cite{MS23} a faster algorithm is given, which runs in $\mathcal{O}(m^2)$.

Because of the inherently dynamic nature of routing games, the Braess Paradox has been studied from several perspectives within game theory also in its dynamic form. Indeed, many aspects of the problem may vary over time. For instance, traffic demand may fluctuate, as it is naturally the case during a typical working day, and the behavior of agents within the network may change accordingly. An interesting question is therefore to understand what happens in networks exhibiting Braess Paradox when the traffic demand varies \cite{CDS24,VC23}, by also allowing the demand to depend explicitly on time and study the resulting evolution of equilibria \cite{DP07,ZC14}.
Another aspect that may change over time is the underlying network itself; in particular, new links may be created or existing ones may be removed, possibly on a very short time scale. 
This is the setting where we place our work.

Indeed, a challenging algorithmic question is how to handle vulnerability in the context of dynamically evolving graphs \cite{DynBook}. Dynamic graph algorithms face problems like graph connectivity \cite{JR87}, the minimum spanning tree and many other graph problems \cite{JTDynProb01}. 
In our context, inter-node connections can change during time, either by appearing or by disappearing and this may consequently change the possibility of having the Braess Paradox. 
We first observe that, as an easy consequence of the characterizations given by  \cite{CGS18} and \cite{ChenEtal15}, 
adding an edge can make vulnerable a non-vulnerable graph, but it never happens that a vulnerable graph becomes non-vulnerable and, dually, removing an edge can make non-vulnerable a vulnerable graph, but not vice versa.

In a companion paper \cite{fgs26a}, 
we give an incremental algorithm to determine if adding an edge to a net makes it vulnerable: focusing on the non-vulnerable case, our dynamic algorithm works on the MIS  
of a non vulnerable graph, adapting the $\mathcal{O}(m^2)$ static algorithm given by \cite{MS23}. 
Such an algorithm has an amortized cost of $\mathcal{O}(m)$. 
%
Here, we address the decremental problem and 
assume that we have a vulnerable net with $n$ nodes and $m$ edges, where some of the latter ones can be sequentially removed from the net. 
In this case, we preliminarily observe that the algorithm given in  \cite{CGS19} is ``intrinsically'' decremental, 
as it iteratively eliminates redundant edges, which do not have impact on vulnerability of the graph.
Indeed, after giving some background in Section~\ref{sec:vuln}, we show in Section \ref{sec:decAlg} that 
any algorithm of complexity $\mathcal{O}(f(n,m))$ following the approach of \cite{CGS19} 
leads to a decremental algorithm of amortized cost of $\mathcal{O}(f(n,m)/m)$ after $\Theta(m)$ edge removals.
If we use the algorithm in \cite{CGS19} as it is, this would lead to a decremental algorithm of amortized cost of  $\mathcal{O}(nm)$, that is   unsatisfactory when compared to the amortized cost $\mathcal{O}(m)$ of the incremental case.
To obtain the desired amortized complexity, we improve the algorithm in \cite{CGS19} and we reduce its complexity from $\mathcal{O}(nm^2)$ to $\mathcal{O}(m^2)$ and, consequently, we also get 
a decremental algorithm of amortized cost $\mathcal{O}(m)$.

To do this, in Section \ref{sec:chars}, we provide
sufficient conditions for either 
recognizing net vulnerability or finding a redundant edge, and in Section \ref{sec:algos} we provide linear-time algorithms for checking such conditions. These algorithms are used in Section~\ref{sec:amortized} to yield the linear-time amortized decremental algorithm. Finally, Section~\ref{sec:concl} concludes the paper, by also pointing to directions for future research. Due to space constraints, some definitions and proofs have been deferred to the Appendices.   



\section{Preliminaries}
\label{sec:vuln}
\vspace{-2mm}    

%
A \emph{directed multigraph} (or  \emph{multi-digraph}) $G = (V,E)$ consists of a
set $V$ of \emph{nodes} (or \emph{vertices}) and a set $E$ of \emph{edges}. Each edge $e$ relates a pair of nodes $(u, v)$. 
%
When the name of an edge is not relevant but only its extremes are,
we denote an edge with the pair it relates.
 
A \emph{path} $P$ is a sequence $u_1 e_1 u_2 \ldots u_{n-1} e_{n-1} u_n$ (for $n \geq 1$) of nodes and edges 
such that $e_i$ relates $(u_i,u_{i+1})$, for all $i<n$;
$u_1$ and $u_n$ are called {\em extremes} and $u_2,\ldots,u_{n-1}$ are called {\em internal nodes}.
If no node appears more than once in $P$, we say that $P$ is {\em simple} (or {\em acyclic}).
In this paper, when we refer to a path we always mean that it is simple, if not explicitly said.

\begin{definition}
Given a path $P$ and two nodes $a,b\in P$, we write $a<_P b$ (or $a<b$ on $P$) if $a$ appears before $b$ in $P$; $\le$ is defined accordingly.
\vspace{-1mm}    
\end{definition}

Notationally, given a path $P$ from $a$ to $b$, we denote it as $P=a\pat b$ or $a\overset{P}{\pat}b$; 
if $c <_P d$, 
we write $c\overset{P}{\pat}d$ to denote the part of $P$ between $c$ and $d$. We use round/square parenthesis to exclude/include extremes; for example, $(a\pat b]$ means that we are excluding $a$ and including $b$. We use $P$ to mean $[a\pat b]$ and $\mathring P$ to mean $(a\pat b)$; two paths $P$ and $Q$ are {\em internally node disjoint} if $\mathring P\cap \mathring Q=\emptyset$.
We can join paths that only share one endpoint as follows: if $P=a\pat b$ and $Q=b\pat c$ are internally node disjoint, then $a\overset{P}{\pat}b\overset{Q}{\pat}c$ indicates the path that starts in $a$, follows $P$ until $b$, and finally follows $Q$ until $c$.


We fix two different nodes: a \emph{source}, written $s$, and a \emph{target}, written $t$.
The resulting triple $(G,s,t)$ is called {\em net}; when $s$ and $t$ are understood, we sometimes
simply write it as $G$ (i.e., without specifying the source and target) and use the term ‘graph' instead of ‘net'.
An \emph{st-path} is a 
path from $s$ to $t$. 
We say that a net is $st$-{\em connected} if every node is touched by at least 
one $st$-path.
%
%

\label{sec:IRRED}

The formal definitions of Braess Paradox and vulnerable nets are deferred to Appendix~\ref{app:braess}, 
since they are not fundamental to appreciate our results. Indeed, for our aims it suffices to know that a net is vulnerable if and only if one of the following two characterizations hold. 
The first characterisation of vulnerable nets holds for {\em irredundant} nets,
defined as those nets containing only irredundant nodes and edges.
As defined in \cite{ChenEtal15} (and derived from \cite{Milch06}), a vertex or an edge is {\em irredundant}
if it appears in a simple $st$-path, and {\em redundant} otherwise. 
The characterisation given by \cite{CinesiFull} states that an irredundant net is vulnerable if and only if it is series-parallel
(the definition is reported in Appendix~\ref{app:braess} and it can be checked in linear time \cite{VTL82}).
However, verifying if a net is irredundant is an \NP-complete problem, as shown in \cite{CGS19}.
In \cite{CGS18} the authors give a second characterisation of vulnerability that holds for generic (i.e., not necessarily irredundant) graphs, based on the embedding of $\W$ (Fig.~\ref{fig:Wheat}(3)) within the given net.
\vspace{-1mm}    

\begin{definition}[\cite{LR80}] 
	A {\em subgraph homeomorphism} from $H$ to $G$
	is a pair of injective mappings $(\phi,\psi)$, from $V_H$ to $V_G$ and
	from $E_H$ to the simple paths of $G$ respectively, 
	such that $\psi(e) = \phi(x) \leadsto \phi(y)$, for every $e \in  E_H$ that relates $(x,y)$.
	A homeomorphism is called {\em node disjoint} if all paths in the image of $\psi$ are pairwise 
	internally node disjoint.
\vspace{-1mm}    
\end{definition}

\begin{definition}
	\label{def:stEmbedding}
	An {\em st-embedding} of $\W$ from Fig.~\ref{fig:Wheat}(3) into a net $(G,s',t')$
	is a node disjoint subgraph homeomorphism $(\phi,\psi)$ from $\W$ to $G$ such that
	there are (possibly empty) node disjoint simple paths from $s'$ to $\phi(s)$ and from $\phi(t)$ to $t'$ that
	are pairwise internally node disjoint with all paths in the image of $\psi$.       
\end{definition}
\section{A High-level Decremental Algorithm}
\label{sec:decAlg}  

The high-level structure of the algorithm given in \cite{CGS19} 
is summarized 
in Algorithm~\ref{alg:cgs}.
It consists in iteratively choosing a cycle $\C$ in the graph $G$ (line \ref{line:findCycle}) and analyzing that cycle (line \ref{line:cycleAnalyis}): this can lead to find a $st$-embedding of the Wheatstone graph $\mathcal{W}$, that implies that $G$ is vulnerable (line \ref{line:findW}), or to break that cycle, by removing 
an edge $e$
redundant in $\C$. In this latter case,  function {\sl rmvRedEdges} is called
(line \ref{line:removing}) to identify the set of all edges that are no more $st$-connected in the graph after removing $e$; these edges now all become redundant and are removed from $G$.  
If there is no embedding of $\mathcal{W}$ involving (parts of) a cycle in the original graph, the {\bf while}-loop stops in at most $\mathcal{O}(m)$ steps; at that moment, $G$ is an acyclic graph and its vulnerability is equivalent to not be series-parallel \cite{CinesiFull} (this can be checked in linear time \cite{VTL82}).
The complexity of this algorithm is $f(n,m)$ $\in\mathcal{O}(m\cdot g(n,m))$, being $g(n,m)$ the complexity of the function {\sl Search$\mathcal{W}$}: indeed, both {\sl FindACycle} and {\sl rmvRedEdges} are linear in $m$, being variations of some visit; moreover,
we expect that $g(n,m)\in\Omega(m)$.

Any algorithm following the same ``hunting $\mathcal{W}$ in cycles'' strategy as function {\sl isVuln} can easily be transformed into a decremental procedure: just invoke it on the graph $G\setminus\{e\}$, being $e$ the edge to be removed. 
However, to carry out our amortized complexity analysis, we found it convenient to represent a sequence of edge removals as an algorithm (Algorithm \ref{alg:dec}), where we assume to know in advance all the edges that will be removed and we specify how to use {\sl isVuln} in this context. 
In practice, edges will disappear in different moments of the net's life and are not available in advance: Algorithm \ref{alg:dec} shows how to exploit the current status of the net ({\sc Vuln/NotVuln}) and the edges identified as redundant in previous executions (stored in $R$) to optimize the calls to {\sl isVuln} and, consequently, reduce the amortized cost.

More precisely, in the decremental setting, we start from a net with $m$ edges and sequentially remove $k\in\mathcal{O}(m)$ edges $e_1, \ldots, e_k$. At any edge removal, if the graph is vulnerable, we call function {\sl IsVuln} on a graph $\bar{G}$ in which we delete the edge to be removed and the set $R$ of all edges already discovered redundant in  previous executions of {\sl IsVuln} (line~\ref{line:seqGbar}--\ref{line:seqIsVuln})  
and we keep track of new  redundant edges discovered 
(line~\ref{line:seqRR'}).  In fact, removing an edge $e_j$ eliminates some $st$-paths in $G$ and therefore all redundant edges in $G$ will remain redundant in $G\setminus\{e_j\}$.
 Thus, during a sequence of edge removals, the set of redundant edges can only grow.
 In contrast, if the net is not vulnerable, thanks to the characterization in \cite{CGS18}, we know that it will remain so after any further edge removal; thus, we can remove edges in constant time, in particular without invoking {\sl IsVuln}.
 \vspace*{-.4cm}
\begin{center}
\begin{tabular}{cc}
\begin{minipage}{0.46\linewidth}
\begin{algorithm}[H]
    \caption{{\sl IsVuln}$(G)$}
    \begin{algorithmic}[1]
        \State $R'$ = $\emptyset$ 
        \While{$G$ is not acyclic}
        \State $\C$ = {\sl FindACycle}$(G)$ \label{line:findCycle}
        \State {\sl found}, $e$ = {\sl Search${\cal W}$}$(G,\C)$ \label{line:cycleAnalyis}
        \If{{\sl found}}
        \State \Return {$R'$, {\sc Vuln}}
        \label{line:findW}
        \Else \State $R'$ = $R'\, \cup $ {\sl rmvRedEdges}$(G,e)$
        \label{line:removing}
        \EndIf
        \EndWhile
        \If{$G$ is series-parallel}
        \State \Return {$\emptyset$, \sc NotVuln}
        \Else \ \Return {$R'$, \sc Vuln}
        \EndIf
    \end{algorithmic}
    \label{alg:cgs}
\end{algorithm}
\end{minipage}
&
\begin{minipage}{0.47\linewidth}
\begin{algorithm}[H]
    \caption{{\sl seqRmv}$(G, e_1, \ldots, e_k)$}
    \begin{algorithmic}[1]
    \State $R$ = $\varnothing$
    \State {\sl status} = {\sc Vuln}
    \For{$j=1$ {\bf to} $k$}
        \If{$e_j\notin R$ and {\sl status} = {\sc Vuln}}
            \State  $\bar{G}=(V, E\setminus (R \cup e_j))$  
        \label{line:seqGbar}
            \State $R'$, {\sl status} = {\sl IsVuln}$(\bar{G})$
        \label{line:seqIsVuln}

            \State{$R=R\cup R'$}
        \label{line:seqRR'}
        \EndIf
        \State{$E = E \setminus \{e_j\}$}
    \EndFor
    \State \Return {\sl status}
    \end{algorithmic}
    \label{alg:dec}
\end{algorithm}
\end{minipage}
\vspace{5mm}
\end{tabular}
\end{center}

If the net is vulnerable, the function {\sl IsVuln}  
has analyzed some cycles before finding a $st$-embedding of $\mathcal{W}$ and has marked a set $R$ of edges as redundant; these can 
be ignored in subsequent calls of {\sl IsVuln} because removing edges never makes irredundant a redundant edge. 
 %
This implies that a sequence of runs of {\sl IsVuln} cannot remove more edges than {\sl IsVuln} does in its worst case scenario, i.e. whenever the algorithm has to break all the existing cycles, with a cost of $\mathcal{O}(m\cdot f(n, m))$.
In other words, $k$ runs of {\sl IsVuln} as a decremental algorithm after $k$ edge removals (supported by storing in $R$ all the redundant edges found) has a total worst case complexity bounded by the very same worst case complexity of {\sl IsVuln} itself. 
The total cost for checking vulnerability after removing $k\in\mathcal{O}(m)$ edges is then $C(m)\in\mathcal{O}(m\cdot f(n,m))$.
Consequently, the amortized cost for the single iteration is $\frac{C(m)}{k}\in\mathcal{O}(f(n,m))$, 
being $k\in\mathcal{O}(m)$.


 In \cite{CGS19}, the algorithm for {\sl Search$\mathcal{W}$} consists of solving a constant number of reachability problems. However, it may happen that, with the chosen cycle, it is not able to either find a $st$-embedding of ${\cal W}$ nor to break that cycle by canceling some redundant edge(s); in that case, a new cycle must be chosen. Following a suitable order on cycles, this situation can occur at most $\mathcal{O}(n)$ times. As a consequence, the complexity of {\sl Search$\mathcal{W}$} in \cite{CGS19} is $\mathcal{O}(nm)$ and the complexity of the overall algorithm is $\mathcal{O}(nm^2)$.
In the rest of this paper, we show that we can refine the cycle analysis presented in \cite{CGS19} and make it always conclusive (i.e. it always either finds ${\cal W}$ or breaks a cycle) in time $\mathcal{O}(m)$, by solving a constant number of reachability problems.
This leads to a function {\sl Search$\mathcal{W}$} of complexity $\mathcal{O}(m)$ (see Algorithm~\ref{alg:CycleAnalysis} in Section~\ref{sec:amortized}), which brings the complexity of {\sl isVuln} to $\mathcal{O}(m^2)$ and allows for a decremental algorithm of amortized cost of $\mathcal{O}(m)$. 


\section{Graph-Theoretic Conditions for Vulnerability}
\label{sec:chars}

The theoretical heart of this paper is the development of 7 sufficient graph-theoretic conditions for establishing vulnerability of a net analyzing a cycle, with the feature that, if all the checks fail, we identify at least an edge that is redundant in that cycle (see Theorem~\ref{th:Main}). To this aim, we start by giving some preliminary definitions and name conventions.


\vspace{-2mm}
\subsection{Preliminary Definitions and Notations}\label{Definitions}

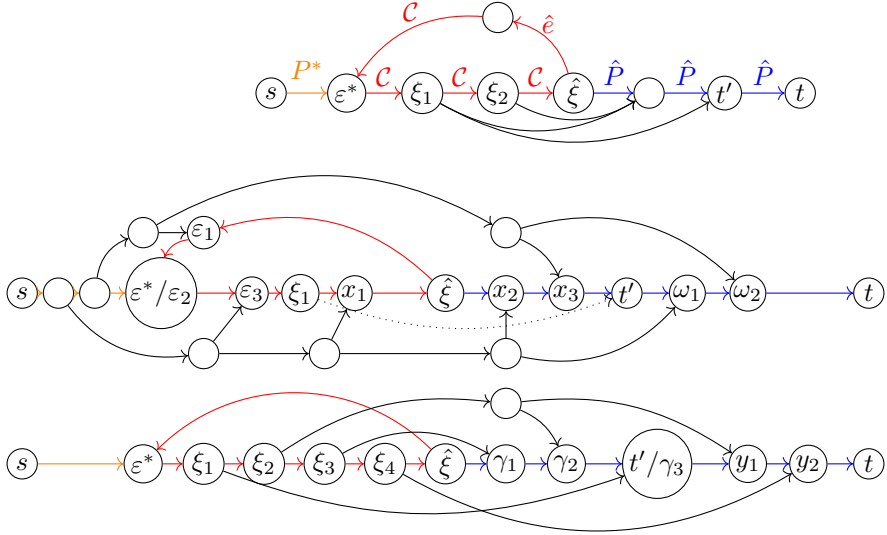
\begin{figure}[t]
\centering
\begin{minipage}{1\textwidth}
\centering
\begin{tikzpicture}[every node/.style={circle,draw, minimum size=4mm, inner sep=1pt}]
\node (s) at (0,0) {$s$};
\node (es) at (1,0) {$\e^*$};
\node (x1) at (2,0) {$\x_1$};
\node (xi) at (3,0) {$\x_2$};
\node (xt) at (4,0) {$\hat \x$};
\node (a) at (5,0) {};
\node (tp) at (6,0) {$t'$};
\node (t) at (7,0) {$t$};
\node (b) at (3,1) {};

\draw [->, orange] (s) -- node[above,draw=none, text=orange,] {$P^*$} (es);
\draw[->,red] (es)to node[above,draw=none, text=red] {$\C$}(x1);
\draw[->,red] (x1)to node[above,draw=none, text=red] {$\C$}(xi);
\draw [->,bend right] (x1) to (a);
\draw [->,bend right] (x1) to (tp);
\draw [->,bend right] (xi) to (a);
\draw[->, red] (xi)to node[above,draw=none, text=red] {$\C$}(xt);
\draw[->,blue] (xt)-- node[above,draw=none, text=blue] {$\hat P$}(a);
\draw [->,red,bend right =30] (xt) to node[above,draw=none, text=red] {$\hat e$}(b);
\draw [->,red,bend right =30] (b) to node[above,draw=none, text=red] {$\C$}(es);
\draw[->,blue] (a)-- node[above,draw=none, text=blue] {$\hat P$}(tp);
\draw[->,blue] (tp)-- node[above,draw=none, text=blue] {$\hat P$}(t);
\end{tikzpicture}
\end{minipage}
\begin{minipage}{1\textwidth}

\begin{tikzpicture}[scale=0.8,every node/.style={circle,draw, minimum size=4mm, inner sep=0.5pt}]
\node (s) at (0,0) {$s$};
\node (a) at (0.6,0) {};
\node (b) at (1.2,0) {};
\node (c) at (2,1) {};
\node (e1) at (3,1) {$\e_1$};
\node (d) at (3,-1) {};
\node (e2) at (3.8,0) {$\e_3$};
\node (e) at (5,-1) {};
\node (f) at (8,-1) {};
\node (g) at (8,1) {};
\node (x1b) at (5.5,0) {$x_1$};
\node (x2) at (8,0) {$x_2$};

\node (x3) at (9,0) {$x_3$};

\node (es) at (2.3,0) {$\e^*/\e_2$};
\node (x1) at (4.6,0) {$\x_1$};
\node (xt) at (7,0) {$\hat \x$};

\node (o1) at (11,0) {$\omega_1$};
\node (tp) at (10,0) {$t'$};
\node (o2) at (12,0) {$\omega_2$};
\node (t) at (14,0) {$t$};

\draw [->,orange] (s) -- (a);
\draw[->,orange] (a) -- (b);
\draw[->,orange] (b) -- (es);
\draw [->,red] (es) -- (e2);
\draw[->,red] (e2) -- (x1);
\draw[->,red] (x1) -- (x1b);
\draw [->,red] (x1b) -- (xt);
\draw[->,blue] (xt) -- (x2);
\draw [->,blue] (x2) -- (x3);
\draw[->, blue] (x3) -- (tp);
\draw[->, blue] (tp) -- (o1);
\draw[->, blue] (o1) -- (o2);
\draw[->, blue] (o2) -- (t);

\draw [->,bend left] (b) to (c);
\draw [->,bend left] (c) to (g);
\draw [->,bend left] (g) to (o2);
\draw [->] (c) to (e1);
\draw [->,bend left=20] (g) to (x3);
\draw [->,red, bend right] (xt) to (e1);
\draw [->,red, bend right] (e1) to (es);
\draw [->,bend right] (a) to (d);
\draw [->] (d) to (e2);
\draw [->] (d) to (e);
\draw [->] (e) to (x1b);
\draw [->] (e) to (f);
\draw [->] (f) to (x2);
\draw [->,bend right] (f) to (o1);
\draw [->,bend right=20, dotted] (x1) to (tp);

\end{tikzpicture}
\end{minipage}
\centering
\begin{minipage}{1\textwidth}

\begin{tikzpicture}[scale=0.8,every node/.style={circle,draw, minimum size=4mm, inner sep=1pt}]
\node (s) at (0,0) {$s$};
\node (a) at (8,1) {};
\node (g1) at (8,0) {$\gamma_1$};

\node (g2) at (9,0) {$\gamma_2$};

\node (es) at (2,0) {$\e^*$};
\node (x1) at (3,0) {$\x_1$};
\node (xi) at (4,0) {$\x_2$};
\node (xj) at (5,0) {$\x_3$};
\node (xk) at (6,0) {$\x_4$};
\node (xt) at (7,0) {$\hat \x$};

\node (y1) at (12,0) {$y_1$};
\node (tp) at (10.5,0) {$t'/\gamma_3$};
\node (y2) at (13,0) {$y_2$};
\node (t) at (14,0) {$t$}; 

\draw [->,orange] (s) -- (es);
\draw [->,red] (es) -- (x1);
\draw[->,red] (x1) -- (xi);
\draw [->,red] (xi) -- (xj);
\draw[->,red] (xj) -- (xk);
\draw[->,red] (xk) -- (xt);
\draw[->,blue] (xt) -- (g1);
\draw [->,blue] (g1) -- (g2);
\draw[->, blue] (g2) -- (tp);
\draw[->, blue] (tp) -- (y1);
\draw[->, blue] (y1) -- (y2);
\draw[->, blue] (y2) -- (t);
\draw[->,red,bend right=45] (xt) to (es);
\draw[->,bend left=20] (xi) to (a);
\draw[->,bend left=20] (a) to (y1);
\draw[->,bend left=20] (a) to (g2);
\draw[->,bend right=20] (x1) to (tp);
\draw[->,bend right=30] (xk) to (y2);
\draw[->,bend left=30] (xj) to (g1);

\end{tikzpicture}
\vspace{-2mm}    
\end{minipage}
\caption{Pictorial representation of the conventions used for node names and path colors. To make this figure (and the following ones) more readable, all arrows apart from $\hat e$ are intended to denote paths (i.e., we use $\rightarrow$ instead of $\pat$).}
\label{ex:prel}
\vspace{-3mm}    
\end{figure}

In the rest of this paper, we shall use the following conventions on node names and on path coloring within figures:
\begin{itemize}
    \item $\e^*$ denotes a node at minimum distance from $s$ that belongs to a non-trivial cycle $\C$; edges of $\C$ will always be depicted in red (see Fig.~\ref{ex:prel} up);
    \item $P^*$ is any shortest path from $s$ to $\e^*$ and it will always be depicted in orange (see Fig.~\ref{ex:prel} up).
\end{itemize}
Once $\e^*$, $\C$ and $P^*$ are chosen, we accordingly define:
\begin{itemize}    
    \item $\Xi$ is the set of nodes from which it is possible to exit $\C$, i.e. is the set of $\x\in \Xi$ such that $\exists P=\x\pat t$ and $\mathring P\cap \C=\emptyset$;
    we number the elements of $\Xi$ as $\x_1, \x_2, \ldots$ by following $[\e^*\overset{\C}{\pat}\e^*)$ and call $\hat \x$ the last one (see Fig.~\ref{ex:prel} up);
    \item $\hat e$ is the edge of $\C$ that starts from $\hat\x$;
    \item $\hat P$ is one of the paths $\hat\x \pat t$ that does not touch $\C$; it will always be denoted in blue (see Fig.~\ref{ex:prel} up);
    \item $t'$ is the node on $\hat P$ closer to $t$ that is reachable from $\x_1$ without passing through other nodes of $\hat P$ and $\C$ (see Fig.~\ref{ex:prel} up);
    \item $\E$ is the set of entry nodes along $(\hat\x\overset{\C}{\pat}\x_1]$, i.e. those nodes $\e$ reachable from $s$ without touching $\hat P$ nor passing through $\C$; we number the elements of $\E$ as $\e_1, \e_2, \ldots$ by following $(\hat\x\overset{\C}{\pat}\x_1]$ (see Fig.~\ref{ex:prel} center); notice that $\e^*$ belongs to this order as well (in Fig.~\ref{ex:prel} center it is $\e_2$);
    \item $X$ is the set of nodes of $(\x_1\overset{\C}{\pat}\hat{\x}\overset{\hat P}{\pat}t')$ reachable from $s$ without passing through other nodes of $\C$ and $\hat P$; we number the elements of $X$ as $x_1, x_2, \ldots$ by following $(\x_1\overset{\C}{\pat}\hat{\x}\overset{\hat P}{\pat}t')$ (see Fig.~\ref{ex:prel} center);
    \item $\Omega$ is the set of nodes of $[t'\overset{\hat P}{\pat}t]$ reachable from $s$ without passing through other nodes of $\hat P$;
    we number the elements of $\Omega$ as $\omega_1, \omega_2, \ldots$ by following $[t'\overset{\hat P}{\pat}t]$ (see Fig.~\ref{ex:prel} center);
    \item $Y$ is the set of nodes of $(t'\overset{\hat P}{\pat}t]$ reachable from $\x_1$ without passing through other nodes of $\hat P$;
    we number the elements of $Y$ as $y_1,y_2, \ldots$ by following $(t'\overset{\hat P}{\pat}t]$ (see Fig.~\ref{ex:prel} bottom);
    \item $\Gamma$ is the set of nodes of $(\hat\x\overset{\hat P}{\pat}t']$ reachable from $\x_1$ without passing through other nodes of $\hat P$;
    we number the elements of $\Gamma$ as $\gamma_1, \gamma_2, \ldots$ by following $(\hat\x\overset{\hat P}{\pat}t']$ (see Fig.~\ref{ex:prel} bottom).
\end{itemize}

\begin{remark}
\label{rem:oneExit}
Notice that, if the cycle $\C$ has only one exit $\xi$, then $\hat e$ is necessarily redundant since, starting from $\x$, there is no simple path passing through $\hat e$ that exits $\C$ without passing through $\x$. Thus, in the following we assume that there are at least two exits (viz., $\x_1$ and $\hat\x$). 
\end{remark}

\vspace{-2mm}    

\subsection{Sufficient Conditions for Vulnerability}
\vspace{-1mm}    

Assume that we fixed $\e^*$, $\C$ and $P^*$ as described in the previous section.
We are now ready to provide the 7 conditions and prove that they all imply vulnerability of the net; in the next section, we shall provide algorithms for checking each of them in time that is linear in the number of edges of the net.

\begin{definition}
Each  condition is a property of the underlying graph $G$:   
\begin{itemize}
    \item $A_1(G) \triangleq$ $G$ has at least one $x_i$
    \item $A_2(G) \triangleq$ $\exists \xi_i,\omega_j$ s.t. $\exists P_i=\x_i\pat t$ in $G\setminus\{\omega_j\}$ 
    \item $A_3(G) \triangleq$ $\e_1\in(\hat\x\overset{\C}{\pat}\e^*)$ and $\exists a\in(\e^*\overset{\C}{\pat}\hat\x]$ s.t. $\exists P=\e_1\pat a$ in $G\setminus\{\e^*\}$
    \item $A_4(G) \triangleq$ $G$ has at least one $y_i$
    \item $A_5(G) \triangleq$ $\exists \gamma_i$ $\exists P=\hat\x\pat t$ s.t. $\mathring P \cap ( \C \cup \{\gamma_i\}) = \emptyset$ 
    \item $A_6(G) \triangleq$ $\exists\x_i,\x_j$ with $i<j$ and
    $\exists \gamma,\gamma'\in (\hat\x\overset{\hat P}{\pat}t]$ such that $\exists P_i=\x_i\pat \gamma$ $\exists  P_j=\x_j\pat \gamma'$ s.t. $\mathring{P_i}\cap(\C\cup\hat P)=\emptyset$ and $\mathring{P_j}\cap(\C\cup\hat P)=\emptyset$,
    for which it holds that
either {\bf (I)} $\gamma <_{\hat P} \gamma'$ 
or {\bf (II)} $\mathring{P_i}\cap\mathring{P_j}\neq\emptyset$
    \item $A_7(G) \triangleq$ $\exists \x_i$ s.t. $\exists P=\e^*\pat\hat\x$ in $G\setminus\{\x_i\}\setminus (\hat\x\overset{\hat P}{\pat}t]$
\end{itemize}
\end{definition}

\begin{proposition}\label{prop:noxi}
If $A_1(G)$ then $G$ is vulnerable.
\end{proposition}
\begin{proof}
If $G$ has at least one $x_i$, it means that there exist a path $P=s\pat x_i$ s.t. $\mathring{P}\cap(\C\cup\hat P)=\emptyset$. 
Moreover $\x_1\notin P$ and $x_i<t'$ on $(\x_1\overset{\C}{\pat}\hat\x\overset{\hat P}{\pat}t')$. 
Let $s'$ be the last intersection along $P$ with $P^*$, and call $P'=s'\overset{P}{\pat} x_i$. 
Let $P_1=\x_1\pat t'$; by minimality of $\e^*$, we have that $P_1\cap P^*=\emptyset$. 
Let $a$ be the first intersection along $P'$ with $P_1$, if any. 

If $x_i\in(\x_1\overset{\C}{\pat}\hat\x]$, then we have two cases:
\begin{enumerate}
\item If $P_1$ and $P'$ are disjoint, we are in the situation depicted to the left and the $st$-embedding of $\W$, needed to prove vulnerability of $G$, can be found as depicted to the right (notice that the paths are all disjoint):

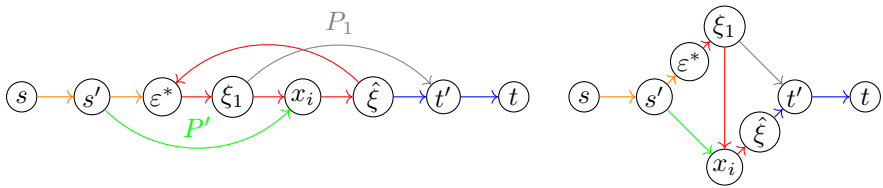
\begin{tikzpicture}[scale=0.93,every node/.style={circle,draw, minimum size=4mm, inner sep=1pt}]
\node (s) at (0,0) {$s$};
\node (sp) at (1,0) {$s'$};
\node (es) at (2,0) {$\e^*$};
\node (x1) at (3,0) {$\x_1$};
\node (xi) at (4,0) {$x_i$};
\node (xt) at (5,0) {$\hat \x$};
\node (tp) at (6,0) {$t'$};
\node (t) at (7,0) {$t$};

\draw [->, orange] (s) -- (sp);
\draw [->, orange] (sp) -- (es);
\draw[->,red] (es) to (x1);
\draw[->,red] (x1) to (xi);
\draw[->, red] (xi) to (xt);
\draw[->,blue] (xt) -- (tp);
\draw [->,red,bend right =45] (xt) to (es);
\draw[->,blue] (tp) -- (t);
\draw [->,green,bend right =45] (sp) to node[above,draw=none, text=green]{$P'$}(xi);
\draw [->,gray,bend left =45] (x1) to node[above,draw=none, text=gray]{$P_1$}(tp);
\begin{scope}[shift={(8,0)}]
\node (s) at (0,0) {$s$};
\node (sp) at (1,0) {$s'$};
\node (es) at (1.5,0.5) {$\e^*$};
\node (x1) at (2,1) {$\x_1$};
\node (xi) at (2,-1) {$x_i$};
\node (xt) at (2.5,-0.5) {$\hat \x$};
\node (tp) at (3,0) {$t'$};
\node (t) at (4,0) {$t$};

\draw [->, orange] (s) -- (sp);
\draw [->, orange] (sp) -- (es);
\draw[->,red] (es) to (x1);
\draw[->,red] (x1) to (xi);
\draw[->, red] (xi) to (xt);
\draw[->,blue] (xt) -- (tp);
\draw[->,blue] (tp) -- (t);
\draw [->,green] (sp) to (xi);
\draw [->,gray] (x1) to (tp);
\end{scope}
\end{tikzpicture}

\item If $P_1$ and $P'$ intersect (and so the) $a$ defined above exists), then:

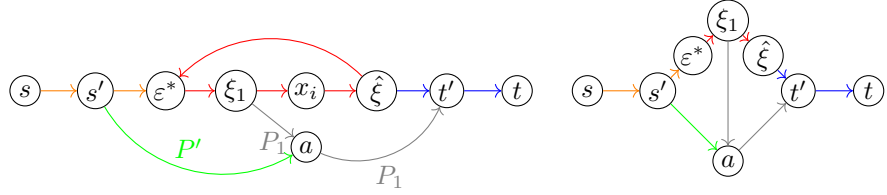
\begin{tikzpicture}[scale=0.93,every node/.style={circle,draw, minimum size=4mm, inner sep=1pt}]
\node (s) at (0,0) {$s$};
\node (sp) at (1,0) {$s'$};
\node (es) at (2,0) {$\e^*$};
\node (x1) at (3,0) {$\x_1$};
\node (a) at (4,-0.8){$a$};
\node (xi) at (4,0) {$x_i$};
\node (xt) at (5,0) {$\hat \x$};
\node (tp) at (6,0) {$t'$};
\node (t) at (7,0) {$t$};

\draw [->, orange] (s) -- (sp);
\draw [->, orange] (sp) -- (es);
\draw[->,red] (es) to (x1);
\draw[->,red] (x1) to (xi);
\draw[->, red] (xi) to (xt);
\draw[->,blue] (xt) -- (tp);
\draw [->,red,bend right =45] (xt) to (es);
\draw[->,blue] (tp) -- (t);
\draw [->,green,bend right =45] (sp) to node[above,draw=none, text=green] {$P'$} (a);
\draw [->,gray] (x1) to node[below,draw=none, text=gray] {$P_1$}(a);
\draw [->,gray,bend right =45] (a) to node[below,draw=none, text=gray] {$P_1$}(tp);
\begin{scope}[shift={(8,0)}]
\node (s) at (0,0) {$s$};
\node (sp) at (1,0) {$s'$};
\node (es) at (1.5,0.5) {$\e^*$};
\node (x1) at (2,1) {$\x_1$};
\node (a) at (2,-1) {$a$};
\node (xt) at (2.5,0.5) {$\hat \x$};
\node (tp) at (3,0) {$t'$};
\node (t) at (4,0) {$t$};

\draw [->, orange] (s) -- (sp);
\draw [->, orange] (sp) -- (es);
\draw[->,red] (es) to (x1);
\draw[->,red] (x1) to (xt);
\draw[->,blue] (xt) -- (tp);
\draw[->,blue] (tp) -- (t);
\draw [->,green] (sp) to (a);
\draw [->,gray] (x1) to (a);
\draw [->,gray] (a) to (tp);
\end{scope}
\end{tikzpicture}
\end{enumerate}
Finally, if $x_i\in[\hat\x\overset{\hat P}{\pat}t')$ we work in the same way; the only difference is that $\hat\x$ now belongs to $\xi_1 \pat x_i$.
\end{proof}
\begin{proposition}\label{prop:onlydirectentry}
\begin{enumerate} 
\item If $A_2(G)$ then $G$ is vulnerable;
\item If $\lnot A_2(G)$ then in $G$ all simple paths $P=\omega_j\pat t$ do not touch $\C$.
\end{enumerate} 
\end{proposition}
\begin{proof}
For the first claim, if such an $\omega_j$ exists, then there is a $P=s\pat\omega_j$ s.t. $\mathring P\cap \hat P=\emptyset$ and $P\cap\C=\emptyset$. 
Let $s'$ be the last intersection between $P^*$ and $P$ along $P$. 
By minimality of $\e^*$, we have that $P_i\cap P^*=\emptyset$ and $(s' \overset P \pat \omega_j]\cap P^*=\emptyset$.

Let $a$ be the first intersection along $P_i$ with $P$, if any; otherwise, let $a$ be the first intersection after $\omega_j$ along $P_i$ with $\hat P$. 
In both cases, let $P_i'=\x_i\overset{P_i}{\pat} a$ and $b$ be the last intersection along $P_i'$ with $[\x_1\overset{\C}{\pat}\hat\x \overset{\hat P}{\pat}\omega_j)$;
let $P_i''=b\overset{P_i}{\pat} a$. 

If $P_i$ and $P$ are not disjoint, we are in the situation to the left\footnote{
In both the following figures, $\xi_i$ is not shown; for sure, it belongs to $[\xi_1 \overset \C \pat \hat\xi]$, but we do not know its position w.r.t. $b$ along $\C$.
}
 and we can find an $st$-embedding of $\W$ as depicted to the right (since the paths are all disjoint):

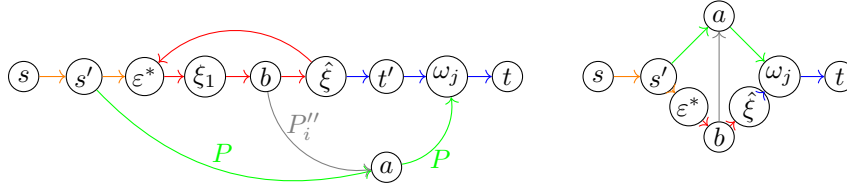
\hspace*{-.4cm}
\begin{tikzpicture}[scale=0.8,every node/.style={circle,draw, minimum size=4mm, inner sep=1pt}]
\node (s) at (0,0) {$s$};
\node (sp) at (1,0) {$s'$};
\node (es) at (2,0) {$\e^*$};
\node (x1) at (3,0) {$\x_1$};
\node (b) at (4,0) {$b$};
\node (xt) at (5,0) {$\hat \x$};
\node (tp) at (6,0) {$t'$};
\node (oi) at (7,0) {$\omega_j$};
\node (a) at (6,-1.5) {$a$};
\node (t) at (8,0) {$t$};

\draw [->, orange] (s) -- (sp);
\draw [->, orange] (sp) -- (es);
\draw[->,red] (es) to (x1);
\draw[->,red] (x1) to (b);
\draw[->, red] (b) to (xt);
\draw[->,blue] (xt) -- (tp);
\draw [->,red,bend right =45] (xt) to (es);
\draw[->,blue] (tp) -- (oi);
\draw[->,blue] (oi) -- (t);
\draw [->,green,bend right] (sp) to node[above,draw=none, text=green]{$P$}(a);
\draw [->,green,bend right =45] (a) to node[below,draw=none, text=green]{$P$}(oi);
\draw [->,gray,bend right =45] (b) to node[above,draw=none, text=gray]{$P_i''$}(a);

\begin{scope}[shift={(9.5,0)}]

\node (s) at (0,0) {$s$};
\node (sp) at (1,0) {$s'$};
\node (es) at (1.5,-0.5) {$\e^*$};
\node (b) at (2,-1) {$b$};
\node (xt) at (2.5,-0.5) {$\hat \x$};
\node (oi) at (3,0) {$\omega_j$};
\node (a) at (2,1) {$a$};
\node (t) at (4,0) {$t$};

\draw [->, orange] (s) -- (sp);
\draw [->, orange] (sp) -- (es);
\draw[->,red] (es) to (b);
\draw[->,red] (b) to (xt);
\draw[->,blue] (xt) -- (oi);
\draw[->,green] (a) -- (oi);
\draw[->,blue] (oi) -- (t);
\draw [->,green] (sp) to (a);
\draw [->,gray] (b) -- (a);
\end{scope}
\end{tikzpicture}

\noindent
Otherwise:

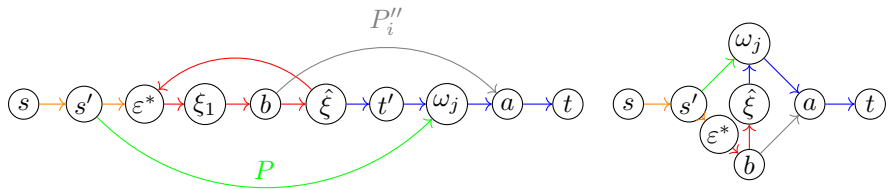
\hspace*{-.4cm}
\begin{tikzpicture}[scale=0.8,every node/.style={circle,draw, minimum size=4mm, inner sep=1pt}]
\node (s) at (0,0) {$s$};
\node (sp) at (1,0) {$s'$};
\node (es) at (2,0) {$\e^*$};
\node (x1) at (3,0) {$\x_1$};
\node (b) at (4,0) {$b$};
\node (xt) at (5,0) {$\hat \x$};
\node (tp) at (6,0) {$t'$};
\node (oi) at (7,0) {$\omega_j$};
\node (a) at (8,0) {$a$};
\node (t) at (9,0) {$t$};

\draw [->, orange] (s) -- (sp);
\draw [->, orange] (sp) -- (es);
\draw[->,red] (es) to (x1);
\draw[->,red] (x1) to (b);
\draw[->, red] (b) to (xt);
\draw[->,blue] (xt) -- (tp);
\draw [->,red,bend right =45] (xt) to (es);
\draw[->,blue] (tp) -- (oi);
\draw[->,blue] (oi) -- (a);
\draw[->,blue] (a) -- (t);
\draw [->,green,bend right =45] (sp) to node[above,draw=none, text=green]{$P$}(oi);
\draw [->,gray,bend left =45] (b) to node[above,draw=none, text=gray]{$P_i''$}(a);
\begin{scope}[shift={(10,0)}]
\node (s) at (0,0) {$s$};
\node (sp) at (1,0) {$s'$};
\node (es) at (1.5,-0.5) {$\e^*$};
\node (b) at (2,-1) {$b$};
\node (xt) at (2,0) {$\hat \x$};
\node (oi) at (2,1) {$\omega_j$};
\node (a) at (3,0) {$a$};
\node (t) at (4,0) {$t$};

\draw [->, orange] (s) -- (sp);
\draw [->, orange] (sp) -- (es);
\draw[->,red] (es) to (b);
\draw[->,red] (b) to (xt);
\draw[->,blue] (xt) -- (oi);
\draw[->,blue] (oi) -- (a);
\draw[->,blue] (a) -- (t);
\draw [->,green] (sp) to (oi);
\draw [->,gray] (b) -- (a);
\end{scope}
\end{tikzpicture}

We then prove the second claim. By contradiction, let $P=\omega_j\pat t$ be a simple path that touches $\C$. 
The only way to exit $\C$ is using some $\x_i$ and, since $P$ is simple, there exists some $Q=\x_i\pat t$ s.t. $\omega_j\notin Q$.
Since $\x_1\overset{\C}{\pat}\x_i$ always exists, this would imply $A_2(G)$, since there would be a path $\x_1\overset{\C}{\pat}\x_i\overset{Q}{\pat}t$ in $G\setminus\{\omega_j\}$.
\end{proof}

\begin{proposition}\label{prop:nobadchord}
If $A_3(G)$ then $G$ is vulnerable.
\end{proposition}
\begin{proof}
Let $P_1=s\pat \e_1$ be such that $P_1\cap\hat P=\emptyset$ and $\mathring P_1\cap\C=\emptyset$ (one surely exists since $\e_1 \in \E$). 
Let $s'$ be the last node of $P^*$ along $P_1$ and $b\in\C$ be the last node of $P$ before $a$ on $[\e_1\overset{\C}{\pat}\e^*)$. 
Also, let $c$ be the last node of $P_1$ along $P$ after $b$ (and before $a$), if any. 

If $b\overset{P}{\pat}a$ and $s'\overset{P_1}{\pat}\e_1$ are disjoint, then we are in the situation depicted at left and we have the $st$-embedding of $\W$ depicted at right:

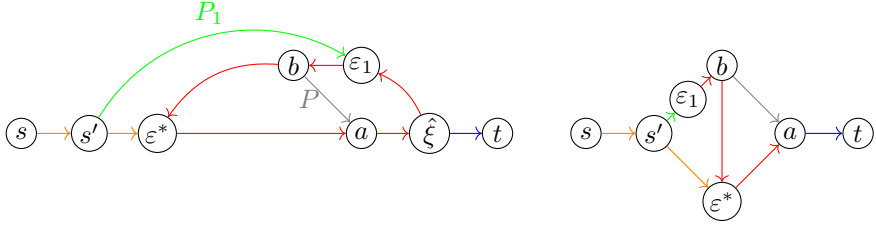
\hspace*{-.4cm}
\begin{tikzpicture}[scale=0.9,every node/.style={circle,draw, minimum size=4mm, inner sep=1pt}]
\node (s) at (0,0) {$s$};
\node (sp) at (1,0) {$s'$};
\node (es) at (2,0) {$\e^*$};
\node (b) at (4,1) {$b$};
\node (e1) at (5,1) {$\e_1$};
\node (a) at (5,0) {$a$};
\node (xt) at (6,0) {$\hat\x$};
\node (t) at (7,0) {$t$};

\draw [->, orange] (s) -- (sp);
\draw [->, orange] (sp) -- (es);
\draw[->,red] (es) to (a);
\draw[->,red] (a) to (xt);
\draw[->,blue] (xt) -- (t);
\draw [->,red,bend right =20] (xt) to (e1);
\draw [->,red] (e1) to (b);
\draw [->,red,bend right =30] (b) to (es);
\draw [->,green,bend left =45] (sp) to node[above,draw=none, text=green]{$P_1$} (e1);
\draw[->,gray] (b) to node[left,draw=none, text=gray]{$P$}(a);

\begin{scope}[shift={(8.3,0)}]
\node (s) at (0,0) {$s$};
\node (sp) at (1,0) {$s'$};
\node (es) at (2,-1) {$\e^*$};
\node (b) at (2,1) {$b$};
\node (e1) at (1.5,0.5) {$\e_1$};
\node (a) at (3,0) {$a$};
\node (t) at (4,0) {$t$};

\draw [->, orange] (s) -- (sp);
\draw [->, orange] (sp) -- (es);
\draw[->,red] (es) to (a);
\draw[->,gray] (b) to (a);
\draw[->,green] (sp) -- (e1);
\draw[->,red] (e1) -- (b);
\draw[->,blue] (a) -- (t);
\draw[->,red] (b) to (es);
\end{scope}
\end{tikzpicture}

\noindent
Otherwise:

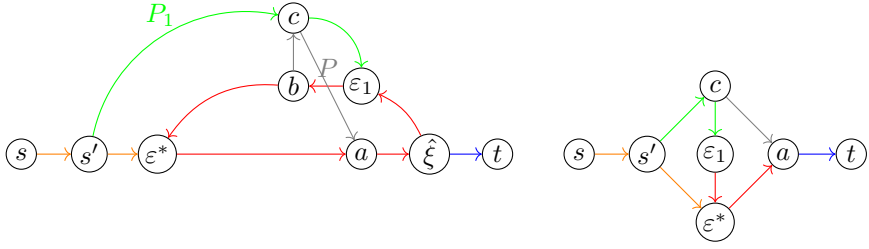
\hspace*{-.4cm}
\begin{tikzpicture}[scale=0.9,every node/.style={circle,draw, minimum size=4mm, inner sep=1pt}]
\node (s) at (0,0) {$s$};
\node (sp) at (1,0) {$s'$};
\node (es) at (2,0) {$\e^*$};
\node (b) at (4,1) {$b$};
\node (c) at (4,2) {$c$};
\node (e1) at (5,1) {$\e_1$};
\node (a) at (5,0) {$a$};
\node (xt) at (6,0) {$\hat\x$};
\node (t) at (7,0) {$t$};

\draw [->, orange] (s) -- (sp);
\draw [->, orange] (sp) -- (es);
\draw[->,red] (es) to (a);
\draw[->,red] (a) to (xt);
\draw[->,blue] (xt) -- (t);
\draw [->,red,bend right =20] (xt) to (e1);
\draw [->,red] (e1) to (b);
\draw [->,red,bend right =30] (b) to (es);
\draw [->,green,bend left =45] (sp) to node[above,draw=none, text=green]{$P_1$}(c);
\draw [->,green,bend left =45] (c) to (e1);
\draw[->,gray, bend left=45] (b) -- (c);
\draw[->,gray] (c) to node[above,draw=none, text=gray]{$P$} (a);
\begin{scope}[shift={(8.2,0)}]
\node (s) at (0,0) {$s$};
\node (sp) at (1,0) {$s'$};
\node (es) at (2,-1) {$\e^*$};
\node (c) at (2,1) {$c$};
\node (e1) at (2,0) {$\e_1$};
\node (a) at (3,0) {$a$};
\node (t) at (4,0) {$t$};

\draw [->, orange] (s) -- (sp);
\draw [->, orange] (sp) -- (es);
\draw[->,red] (es) to (a);
\draw[->,gray] (c) to (a);
\draw[->,green] (sp) -- (c);
\draw[->,red] (e1) -- (es);
\draw[->,blue] (a) -- (t);
\draw[->,green] (c) to (e1);
\end{scope}
\end{tikzpicture}

\vspace{-.2cm}
\end{proof}

\begin{proposition}\label{prop:noyi}
If $A_4(G)$ then $G$ is vulnerable.
\vspace{-2mm}
\end{proposition}
\begin{proof}
If $G$ has at least one $y_i$, then there exists a $P=\x_1\pat y_i$ with $\mathring P\cap \hat P=\emptyset$. 
Let $\x_j$ be the last node of $\C$ along $P$; we define $P'=\x_j\overset{P}{\pat}y_i$ and $P_1=\x_1\pat t'$. 
Since $t' <_{\hat P} y_i$ and there is no $Q=\x_1\pat y_i$ s.t. $\mathring{Q}\cap C=\emptyset$ (by maximality of $t'$), then $P'\cap P_1=\emptyset$, 
$(\mathring P'\cup \mathring P_1)\cap P^*=\emptyset$ and $(\mathring P'\cup \mathring P_1)\cap \hat P=\emptyset$. 
Thus, we are in the situation depicted at left and we can find an $st$-embedding of $\W$ as depicted to the right:

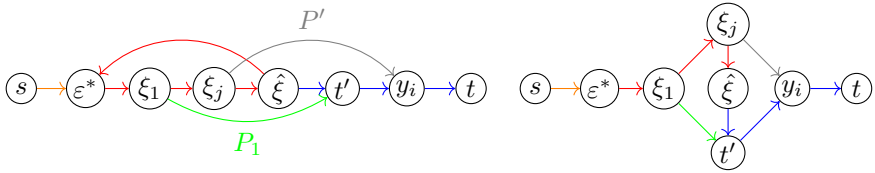
\hspace*{-.4cm}
\begin{tikzpicture}[scale=0.85,every node/.style={circle,draw, minimum size=4mm, inner sep=1pt}]
\node (s) at (0,0) {$s$};
\node (es) at (1,0) {$\e^*$};
\node (x1) at (2,0) {$\x_1$};
\node (xj) at (3,0) {$\x_j$};
\node (xt) at (4,0) {$\hat \x$};
\node (tp) at (5,0) {$t'$};
\node (yi) at (6,0) {$y_i$};
\node (t) at (7,0) {$t$};

\draw [->, orange] (s) -- (es);
\draw[->,red] (es) to (x1);
\draw[->,red] (x1) to (xj);
\draw[->, red] (xj) to (xt);
\draw[->,blue] (xt) -- (tp);
\draw [->,red,bend right =45] (xt) to (es);
\draw[->,blue] (tp) -- (yi);
\draw[->,blue] (yi) -- (t);
\draw [->,green,bend right] (x1) to node[below,draw=none, text=green]{$P_1$}(tp);
\draw [->,gray,bend left =45] (xj) to node[above,draw=none, text=gray]{$P'$}(yi);
\begin{scope}[shift={(8,0)}]
\node (s) at (0,0) {$s$};
\node (es) at (1,0) {$\e^*$};
\node (x1) at (2,0) {$\x_1$};
\node (xt) at (3,0) {$\hat\x$};
\node (xj) at (3,1) {$\x_j$};
\node (tp) at (3,-1) {$t'$};
\node (yi) at (4,0) {$y_i$};
\node (t) at (5,0) {$t$};

\draw [->, orange] (s) -- (es);
\draw[->,red] (es) to (x1);
\draw[->,red] (x1) to (xj);
\draw[->,blue] (xt) -- (tp);
\draw[->,green] (x1) -- (tp);
\draw[->,blue] (yi) -- (t);
\draw [->,red] (xj) to (xt);
\draw [->,blue] (tp) to (yi);
\draw [->,gray] (xj) -- (yi);
\end{scope}
\end{tikzpicture}

\end{proof}

\begin{proposition}\label{prop:jumponpt}
If $A_5(G)$ then $G$ is vulnerable.
\end{proposition}
\begin{proof}
Since a $\gamma_i$ exists, then there exists a $P_i=\x_j\pat\gamma_i$ s.t. $\mathring P_i\cap \C=\emptyset$. 
By minimality of $\e^*$, it holds that $(P\cup P_i)\cap P^*=\emptyset$. 
Let $a$ be the first node of $P$ along $P_i$, if any; otherwise, let $a$ be the first node of $P$ along $\hat P$ (notice that $a>_{\hat P} \gamma_i$ -- at most $a$ is $t$).
Let $P'=b\overset{P}{\pat} a$, where $b$ is the last node of $P$ along $\hat P$ before $a$; thus, $\mathring P'\cap(\hat P\cup P_i)$. 

Since $\mathring P'\cap P_i=\emptyset$ and $\mathring P'\cap \hat P=\emptyset$, we can find an $st$-embedding of $\W$. 
The case in which $b\overset{P'}{\pat}a$ and $P_i$ are disjoint is given to the left and the $st$-embedding of $\W$ at the right:

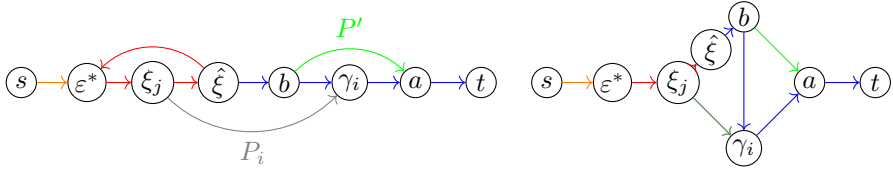
\hspace*{-.4cm}
\begin{tikzpicture}[scale=0.869,every node/.style={circle,draw, minimum size=4mm, inner sep=1pt}]
\node (s) at (0,0) {$s$};
\node (es) at (1,0) {$\e^*$};
\node (xj) at (2,0) {$\x_j$};
\node (xt) at (3,0) {$\hat \x$};
\node (b) at (4,0) {$b$};
\node (gi) at (5,0) {$\gamma_i$};
\node (a) at (6,0) {$a$};
\node (t) at (7,0) {$t$};

\draw [->, orange] (s) -- (es);
\draw[->,red] (es) to (xj);
\draw[->,red] (xj) to (xt);
\draw[->,blue] (xt) -- (b);
\draw [->,red,bend right =45] (xt) to (es);
\draw[->,blue] (b) -- (gi);
\draw[->,blue] (gi) -- (a);
\draw[->,blue] (a) -- (t);
\draw [->,green,bend left=45] (b) to node[above,draw=none, text=green]{$P'$} (a);
\draw [->,gray,bend right =45] (xj) to node[below,draw=none, text=gray]{$P_i$} (gi);
\begin{scope}[shift={(8,0)}]
\node (s) at (0,0) {$s$};
\node (es) at (1,0) {$\e^*$};
\node (xj) at (2,0) {$\x_j$};
\node (xt) at (2.5,0.5) {$\hat\x$};
\node (b) at (3,1) {$b$};
\node (gi) at (3,-1) {$\gamma_i$};
\node (a) at (4,0) {$a$};
\node (t) at (5,0) {$t$};

\draw [->, orange] (s) -- (es);
\draw[->,red] (es) to (xj);
\draw[->,blue] (xt) -- (b);
\draw[->,green] (xj) -- (gi);
\draw[->,blue] (gi) -- (a);
\draw [->,red] (xj) to (xt);
\draw [->,blue] (b) to (gi);
\draw [->,gray] (xj) -- (gi);
\draw [->,green] (b) -- (a);
\draw [->,blue] (a) to (t);
\end{scope}
\end{tikzpicture}

\noindent
The case in which $b\overset{P'}{\pat}a$ and $P_i$ intersect is given below:

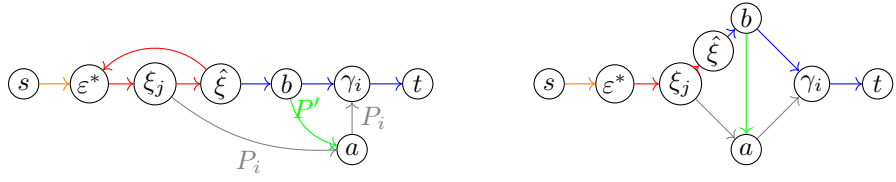
\hspace*{-.4cm}
\begin{tikzpicture}[scale=0.869,every node/.style={circle,draw, minimum size=4mm, inner sep=1pt}]
\node (s) at (0,0) {$s$};
\node (es) at (1,0) {$\e^*$};
\node (xj) at (2,0) {$\x_j$};
\node (xt) at (3,0) {$\hat \x$};
\node (b) at (4,0) {$b$};
\node (gi) at (5,0) {$\gamma_i$};
\node (a) at (5,-1) {$a$};
\node (t) at (6,0) {$t$};

\draw [->, orange] (s) -- (es);
\draw[->,red] (es) to (xj);
\draw[->,red] (xj) to (xt);
\draw[->,blue] (xt) -- (b);
\draw [->,red,bend right =45] (xt) to (es);
\draw[->,blue] (b) -- (gi);
\draw[->,blue] (gi) -- (t);
\draw [->,green,bend right=30] (b) to node[above,draw=none, text=green]{$P'$}(a);
\draw [->,gray,bend right =20] (xj) to node[below,draw=none, text=gray]{$P_i$}(a);
\draw [->,gray] (a) to node[right,draw=none, text=gray]{$P_i$}(gi);
\begin{scope}[shift={(8,0)}]
\node (s) at (0,0) {$s$};
\node (es) at (1,0) {$\e^*$};
\node (xj) at (2,0) {$\x_j$};
\node (xt) at (2.5,0.5) {$\hat\x$};
\node (b) at (3,1) {$b$};
\node (gi) at (4,0) {$\gamma_i$};
\node (a) at (3,-1) {$a$};
\node (t) at (5,0) {$t$};

\draw [->, orange] (s) -- (es);
\draw[->,red] (es) to (xj);
\draw[->,blue] (xt) -- (b);
\draw [->,gray] (a) to (gi);
\draw [->,red] (xj) to (xt);
\draw [->,blue] (b) to (gi);
\draw [->,gray] (xj) -- (a);
\draw [->,green] (b) -- (a);
\draw [->,blue] (gi) to (t);
\end{scope}
\end{tikzpicture}

\noindent
Note that $b\overset{\hat P}{\pat}\gamma_i$, $a\overset{P_i}{\pat}\gamma_i$ (in the intersection case) and $\gamma_i\overset{\hat P}{\pat} a$ (in the disjoint case) are all not empty, since $\gamma_i\notin P$.
\end{proof}

\begin{proposition}\label{prop:orderedexits}
If $A_6(G)$ then $G$ is vulnerable.
\end{proposition}
\begin{proof}
By minimality of $\e^*$, it holds that $(P_i\cup P_j)\cap P^*=\emptyset$. Let $c\in\mathring{P_i}\cap\mathring{P_j}$ be the first node of $P_j$ along $P_i$, if any. 
We show the $st$-embedding of $\W$ for both cases; note that all the paths considered are disjoint by definition.

If $\mathring{P_i}\cap\mathring{P_j}=\emptyset$, then:

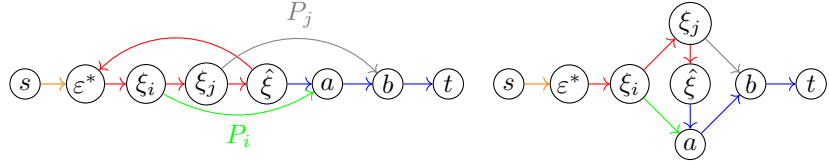
\begin{tikzpicture}[scale=0.8,every node/.style={circle,draw, minimum size=4mm, inner sep=1pt}]
\node (s) at (0,0) {$s$};
\node (es) at (1,0) {$\e^*$};
\node (xi) at (2,0) {$\x_i$};
\node (xj) at (3,0) {$\x_j$};
\node (xt) at (4,0) {$\hat \x$};
\node (a) at (5,0) {$a$};
\node (b) at (6,0) {$b$};
\node (t) at (7,0) {$t$};

\draw [->, orange] (s) -- (es);
\draw[->,red] (es) to (xi);
\draw[->,red] (xi) to (xj);
\draw[->, red] (xj) to (xt);
\draw[->,blue] (xt) -- (a);
\draw [->,red,bend right =45] (xt) to (es);
\draw[->,blue] (a) -- (b);
\draw[->,blue] (b) -- (t);
\draw [->,green,bend right] (xi) to node[below,draw=none, text=green]{$P_i$}(a);
\draw [->,gray,bend left =45] (xj) to node[above,draw=none, text=gray]{$P_j$}(b);
\begin{scope}[shift={(8,0)}]
    
\node (s) at (0,0) {$s$};
\node (es) at (1,0) {$\e^*$};
\node (xi) at (2,0) {$\x_i$};
\node (xt) at (3,0) {$\hat\x$};
\node (xj) at (3,1) {$\x_j$};
\node (a) at (3,-1) {$a$};
\node (b) at (4,0) {$b$};
\node (t) at (5,0) {$t$};

\draw [->, orange] (s) -- (es);
\draw[->,red] (es) to (xi);
\draw[->,red] (xi) to (xj);
\draw[->,blue] (xt) -- (a);
\draw[->,green] (xi) -- (a);
\draw[->,blue] (b) -- (t);
\draw [->,red] (xj) to (xt);
\draw [->,blue] (a) to (b);
\draw [->,gray] (xj) -- (b);
\end{scope}
\end{tikzpicture}

\noindent
otherwise:

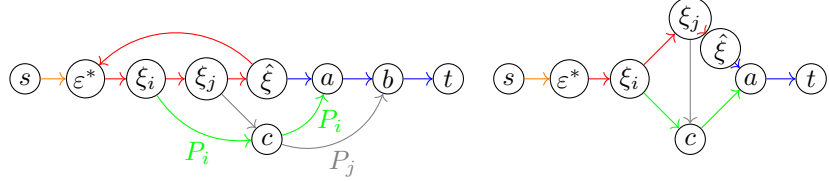
\begin{tikzpicture}[scale=0.8,every node/.style={circle,draw, minimum size=4mm, inner sep=1pt}]
\node (s) at (0,0) {$s$};
\node (es) at (1,0) {$\e^*$};
\node (xi) at (2,0) {$\x_i$};
\node (xj) at (3,0) {$\x_j$};
\node (xt) at (4,0) {$\hat \x$};
\node (c) at (4,-1) {$c$};
\node (a) at (5,0) {$a$};
\node (b) at (6,0) {$b$};
\node (t) at (7,0) {$t$};

\draw [->, orange] (s) -- (es);
\draw[->,red] (es) to (xi);
\draw[->,red] (xi) to (xj);
\draw[->, red] (xj) to (xt);
\draw[->,blue] (xt) -- (a);
\draw [->,red,bend right =45] (xt) to (es);
\draw[->,blue] (a) -- (b);
\draw[->,blue] (b) -- (t);
\draw [->,green,bend right] (xi) to node[below,draw=none, text=green]{$P_i$}(c);
\draw [->,green,bend right] (c) to node[right,draw=none, text=green]{$P_i$\ }(a);
\draw [->,gray] (xj) to (c);
\draw [->,gray,bend right =45] (c) to node[below,draw=none, text=gray]{$P_j$}(b);
\begin{scope}[shift={(8,0)}]
\node (s) at (0,0) {$s$};
\node (es) at (1,0) {$\e^*$};
\node (xi) at (2,0) {$\x_i$};
\node (xt) at (3.5,0.5) {$\hat\x$};
\node (xj) at (3,1) {$\x_j$};
\node (a) at (4,0) {$a$};
\node (c) at (3,-1) {$c$};
\node (t) at (5,0) {$t$};

\draw [->, orange] (s) -- (es);
\draw[->,red] (es) to (xi);
\draw[->,red] (xi) to (xj);
\draw[->,blue] (xt) -- (a);
\draw[->,green] (xi) -- (c);
\draw[->,green] (c) -- (a);
\draw [->,red] (xj) to (xt);
\draw [->,blue] (a) to (t);
\draw [->,gray] (xj) -- (c);
\end{scope}
\end{tikzpicture}

\vspace{-2mm}
\end{proof}

\begin{proposition}\label{prop:noexitjumponC}
If $A_7(G)$ then $G$ is vulnerable.
\end{proposition}
\begin{proof}
Let $a\in[\e^*\overset{\C}{\pat}\x_i)$ be the last node of $P$ before touching the first node $b\in(\x_i\overset{\C}{\pat}\hat\x]$. 
Let $P_i=\x_i\pat t$ be the exit path from $\x_i$ and $c$ be the first intersection of $P_i$ along $\hat P$. 
Let $d$ be the first intersection of $a\overset{P}{\pat}b$ along $\x_i\overset{P_i}{\pat}c$, if any. 

If $a\overset{P}{\pat}b$ and $\x_i\overset{P_i}{\pat}c$ are disjoint, then:

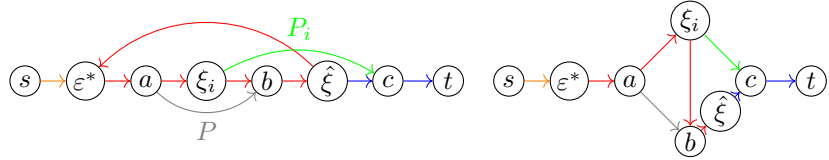
\begin{tikzpicture}[scale=0.8,every node/.style={circle,draw, minimum size=4mm, inner sep=1pt}]
\node (s) at (0,0) {$s$};
\node (es) at (1,0) {$\e^*$};
\node (a) at (2,0) {$a$};
\node (xi) at (3,0) {$\x_i$};
\node (b) at (4,0) {$b$};
\node (xt) at (5,0) {$\hat \x$};
\node (c) at (6,0) {$c$};
\node (t) at (7,0) {$t$};

\draw [->, orange] (s) -- (es);
\draw[->,red] (es) to (a);
\draw[->,red] (a) to (xi);
\draw[->, red] (xi) to (b);
\draw[->, red] (b) to (xt);
\draw[->,blue] (xt) -- (c);
\draw [->,red,bend right =45] (xt) to (es);
\draw[->,blue] (c) -- (t);
\draw [->,green,bend left] (xi) to node[above,draw=none, text=green]{$P_i$}(c);
\draw [->,gray,bend right =45] (a) to node[below,draw=none, text=gray]{$P$}(b);
\begin{scope}[shift={(8,0)}]

\node (s) at (0,0) {$s$};
\node (es) at (1,0) {$\e^*$};
\node (a) at (2,0) {$a$};
\node (xt) at (3.5,-0.5) {$\hat\x$};
\node (xi) at (3,1) {$\x_i$};
\node (b) at (3,-1) {$b$};
\node (c) at (4,0) {$c$};
\node (t) at (5,0) {$t$};

\draw [->, orange] (s) -- (es);
\draw[->,red] (es) to (a);
\draw[->,red] (a) to (xi);
\draw[->,blue] (xt) -- (c);
\draw[->,green] (xi) -- (c);
\draw[->,red] (b) -- (xt);
\draw [->,red] (xi) to (b);
\draw [->,gray] (a) -- (b);
\draw[->,blue] (c) -- (t);
\end{scope}
\end{tikzpicture}

\noindent
otherwise:

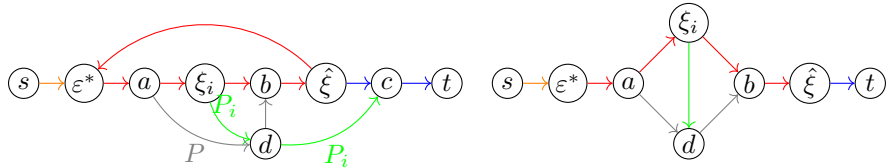
\hspace*{-.4cm}
\begin{tikzpicture}[scale=0.8,every node/.style={circle,draw, minimum size=4mm, inner sep=1pt}]
\node (s) at (0,0) {$s$};
\node (es) at (1,0) {$\e^*$};
\node (a) at (2,0) {$a$};
\node (xi) at (3,0) {$\x_i$};
\node (b) at (4,0) {$b$};
\node (d) at (4,-1) {$d$};
\node (xt) at (5,0) {$\hat \x$};
\node (c) at (6,0) {$c$};
\node (t) at (7,0) {$t$};

\draw [->, orange] (s) -- (es);
\draw[->,red] (es) to (a);
\draw[->,red] (a) to (xi);
\draw[->, red] (xi) to (b);
\draw[->, red] (b) to (xt);
\draw[->,blue] (xt) -- (c);
\draw [->,red,bend right =45] (xt) to (es);
\draw[->,blue] (c) -- (t);
\draw [->,green,bend right] (xi) to node[above,draw=none, text=green]{$P_i$}(d);
\draw [->,green,bend right] (d) to node[below,draw=none, text=green]{$P_i$}(c);
\draw [->,gray,bend right =30] (a) to node[below,draw=none, text=gray]{$P$}(d);
\draw [->,gray] (d) to (b);
\begin{scope}[shift={(8,0)}]
\node (s) at (0,0) {$s$};
\node (es) at (1,0) {$\e^*$};
\node (a) at (2,0) {$a$};
\node (xt) at (5,0) {$\hat\x$};
\node (xi) at (3,1) {$\x_i$};
\node (b) at (4,0) {$b$};
\node (d) at (3,-1) {$d$};
\node (t) at (6,0) {$t$};

\draw [->, orange] (s) -- (es);
\draw[->,red] (es) to (a);
\draw[->,red] (a) to (xi);

\draw[->,green] (xi) -- (d);
\draw[->,red] (b) -- (xt);
\draw [->,red] (xi) to (b);
\draw [->,gray] (a) -- (d);
\draw [->,gray] (d) -- (b);
\draw[->,blue] (xt) -- (t);
\end{scope}
\vspace*{-.8cm}
\end{tikzpicture}

\end{proof}


\vspace{-8mm}
\subsection{Main Theorem}

\begin{theorem}\label{th:Main}
If no $A_i(G)$ holds, for $i = 1,\ldots,7$, then $\hat{e}$ is redundant in $G$.
\end{theorem}
\begin{proof}
    By contradiction, assume that $\hat{e}$ is irredundant; so there exists a (simple) path $P=s\pat \hat\x\pat \x_i\pat t$, for some $i$, such that $\hat e$ is the first edge in $\hat\x\overset P \pat \x_i$. 
    Let $\e_p$ be the first intersection along $P$ with $\C$; we call $P_1=s\overset{P}{\pat}\e_p$, $P_2=\e_p\overset{P}{\pat}\hat\x$ and $P_3=\hat\x\overset{P}{\pat}t$.  
It holds that $P_1\cap\hat P=\emptyset$ since:
    \begin{itemize}
        \item $P_1\cap[\hat\x\pat t')=\emptyset$, otherwise $\e_p$ would be an $x_i$ (and so $A_1(G)$ would hold);
        \item $P_1\cap[t'\pat t]=\emptyset$, otherwise $P_1\cap[t'\pat t]=\omega_j$, for some $j$; by Proposition~\ref{prop:onlydirectentry}(2) there cannot exist an acyclic $\omega_j \pat t$ that touches $\C$, 
        in contradiction with acyclicity of $P$.
    \end{itemize}

 Because $\neg A_1(G)$, we have that $\e_p\notin(\x_1\overset{\C}{\pat}\hat\x]$;
    moreover, we know that $\e_p\in\E$. Thus, $\e_p\in[\hat\xi \overset{\C}{\pat}\x_1]$; by definition of $\e_1$, we have that $\e_p\in[\e_1\overset{\C}{\pat}\x_1]$.
    If $\e_p\in[\e_1\overset{\C}{\pat}\e^*)$, all simple paths $\e_p\pat\hat{\x}$ would pass through $\e^*$, by $\neg A_3(G)$; in that case, we can set $\e_p = \e^*$ and let $P^*$ play the role of $P_1$, and all what follows still works, since $P_2$ and $P_3$ are disjoint from $P^*$, by minimality of $\e^*$. Hence, let us now consider the case in which $\e_p\in[\e^*\overset{\C}{\pat}\x_1]$. 

    Since $\e_p$ is the first node of $P_2$, by $\neg A_7(G)$ we know that $P_2\cap\hat P \neq \emptyset$; let $\gamma_k$ be the first intersection along $P_2$. 
    In particular, by $\neg A_4(G)$, we know that $\gamma_k \in(\hat\x\overset{\hat P}{\pat}t']$.
    Since $\x_i\notin P_2$ (indeed, it belongs to $P_3$), we consider $\bar\x$, the last node of $\C$ touched by $P_2$ before $\gamma_k$. 
    By $\neg A_7(G)$, we have that $\bar{\x}\in[\xi_1\overset{\C}{\pat}\x_i)$ and all the $\x_j \in[\e_p\overset{\C}{\pat}\bar\x]$ are in $P_2$.

    By $\neg A_6(G)$, all paths $P'$ from $\x_i$ to $t$ that do not contain $\gamma_k$ must be such that $\mathring P' \cap [\x_i\overset{\C \cup \hat P}{\pat} \gamma_k] \neq \emptyset$.
Fix one of these $P'$ and let $a$ be the first internal intersection along it with $[\x_i\overset{\C \cup \hat P}{\pat} \gamma_k]$; clearly,  $\hat\x <_{\hat P} a <_{\hat P} \gamma_k$. This however implies $A_5(G)$, since $\hat\xi \overset{\hat P}\pat a \overset{P'}{\pat} t$ is a path 
that does not contain $\gamma_k$. 
Contradiction.
\end{proof}

\section{Algorithmic Checking the  Sufficient Conditions}
\label{sec:algos}


In this section we present the linear-time algorithms that check the conditions $A_1,...,A_7$. 
In what follows, when we perform a visit, we say that a node is a \textit{sink} if we ignore all the edges that leave that node 
and we force to backtrack from that node.

To make the presentation lighter, we present four algorithms that, together, will contribute to the cycle analysis and to the final outcome (i.e., the graph is vulnerable or  a redundant edge is found). However, in doing this, we need to classify nodes and edges as described in Section~\ref{Definitions}, label them accordingly and keep this labeling throughout the execution of these various algorithms. Hence, the graph $G$ is a sort of ‘‘global'' variable, that different algorithms access by reading and modifying the labels of nodes and edges.

For space reasons, all the proofs of this sections are deferred to Appendix~\ref{app:proofs}.
\vspace{-2mm}


    \begin{algorithm}[t]
    \caption{{\sl FindNodes}$(G)$}
       \label{alg:Findthenodes}
    \begin{algorithmic}[1]
        \Require $\e^*$ and all the nodes and edges of $\C$ and $P^*$ are labelled in $G$
        \State 
        Compute the set $\Xi$ of exits, by performing a backward visit from $t$ that considers all nodes in $\C$ as sinks
        \If{$|\Xi| = 1$} \Return $\hat e$ /* $\hat e$ is redundant */
        \EndIf
        \State Sequentially order $\Xi$ by performing a visit from $\e^*$ along $\C$ and label each of its nodes as $\x_1, \x_2, \ldots, \hat\x$
        \State Perform a visit from $\hat\x$ considering all the nodes labelled with $\C$ as sinks and find a simple path between $\hat\x$ and $t$; label all the nodes and edges of this path with $\hat P$
        \State Perform a visit from $\x_1$ considering all nodes labelled with either $\hat P$ or $\C$ as sinks; the last node reached on $\hat P$ (the closest to $t$) is labelled with $t'$
        \ForAll{node $v$ visited in a visit started from $s$}
        \If{$v \in (\x_1\overset{\C}{\pat}\hat{\x}\overset{\hat P}{\pat}t')$}
         \Return {\sc Vuln} /* $A_1(G)$ holds */
        \EndIf
        \If{$v \in [t'\overset{\hat P}{\pat}t]$}
        consider $v$ a sink and add it to $\Omega$
        \EndIf
        \If{$v \in \C\setminus(\x_1\overset{\C}{\pat}\hat\x]$}
        consider $v$ a sink and add it to $\E$
        \EndIf
        \EndFor
        \ForAll{node $v$ visited in a visit started from $\x_1$}
        \If{$v \in (t'\overset{\hat P}{\pat}t]$}
        \Return {\sc Vuln}
        /* $A_4(G)$ holds */
        \EndIf
        \If{$v \in (\hat\x\overset{\hat P}{\pat}t']$}
        consider $v$ a sink and add it to $\Gamma$
        \EndIf
        \EndFor
        \State Sequentially order $\Omega$ by performing a visit from $t'$ along $[t'\overset{\hat P}{\pat}t]$ and label each of its nodes as $\omega_1, \omega_2, \ldots$
        \State Sequentially order $\E$ by performing a visit from $\hat\x$ along $(\hat\x\overset{\C}{\pat}\x_1]$ and label each of its nodes as $\e_1,\e_2, \ldots$
        \State Sequentially order $\Gamma$ by performing a visit from $\hat\x$ along $(\hat\x\overset{\hat P}{\pat}t']$ and label each of its nodes as $\gamma_1, \gamma_2, \ldots$
    \end{algorithmic}
\end{algorithm}

\subsection{\texorpdfstring{Checking Conditions $A_1$ and $A_4$}{Checking Conditions A1 and A4}}
\vspace{-1mm}

Algorithm~\ref{alg:Findthenodes} receives a graph $G$, where some nodes and edges have already been labeled as $\e^*$, $\C$ and $P^*$, and it adds some extra labels to identify all the nodes and paths listed in Section~\ref{Definitions} (notice that some nodes may have more than one label). 
In doing this, it first checks if there are at least two exits (line 2), otherwise it returns that $\hat e$ is redundant (see Remark~\ref{rem:oneExit}).
Also, it checks $A_1$ and $A_4$, by returning {\sc Vuln} if any $x_i$ or $y_i$ is found (see lines 7 and 11, respectively). 
Clearly, the cost is $\mathcal{O}(m)$, since we are just performing a number of visits, and it correctly labels nodes according to the name conventions given in Section~\ref{Definitions}.

\begin{proposition}\label{prop:A1andA4}
{\sl FindNodes}$(G)$ returns {\sc Vuln} if and only if $A_1(G) \lor A_4(G)$.
\vspace{-2mm}       
\end{proposition}


\subsection{\texorpdfstring{Checking Condition $A_3$}{Checking Condition A3}}

\begin{algorithm}[t]
    \caption{{\sl NoBadChord}$(G)$}
    \label{alg:nobadchord}
    \begin{algorithmic}[1]
        \Require In $G$ all nodes and edges are labelled as described in Section~\ref{Definitions}
        \If {$\e_1\in [\e^*\overset{\C}{\pat}\hat\x]$} \Return {\sc NotVuln}
        \EndIf
        \ForAll{node $v$ visited in a visit started from $\e_1$ in $G\setminus\{\e^*\}$}
        \If {$v \in (\e^*\overset{\C}{\pat}\hat\x]$} \Return {\sc Vuln} /* $A_3(G)$ holds */
        \EndIf
        \EndFor
        \State\Return  {\sc NotVuln}
    \end{algorithmic}
\end{algorithm}

Algorithm~\ref{alg:nobadchord} checks condition $A_3$ by first checking whether $\e_1\in [\e^*\overset{\C}{\pat}\hat\x]$ and, if not, by performing a single visit of $G\setminus\{\e^*\}$ from $\e_1$. Its cost is clearly $\mathcal{O}(m)$.

\begin{proposition}\label{prop:A3}
{\sl NoBadChord}$(G)$ returns {\sc Vuln} if and only if $A_3(G)$.
\vspace{-1mm}
\end{proposition}


\vspace{-2mm}       
\subsection{\texorpdfstring{Checking Conditions $A_2$, $A_5$ and $A_7$}{Checking Conditions A2, A5 and A7}}    

Algorithm~\ref{alg:nojumps} takes in input a graph $H$ and some of its nodes $a_0,\ldots,a_k$ that lie along the same simple path that ends in node $a$
and checks whether there is a path from some $a_i$ to $a$ that does not touch some other $a_j$. According to how we instantiate $H$, $a_0,\ldots,a_k$ and $a$, we will be able to check conditions $A_2$, $A_5$ and $A_7$.
Notice that the cost of the algorithm is $\mathcal{O}(m)$, since we do not delete the flag to the visited edges: we are, in fact, performing one single visit of $H$.

\begin{algorithm}[t]
    \caption{{\sl NoJumps}$(H,a_0,\ldots,a_k,a)$}
    \label{alg:nojumps}
    \begin{algorithmic}[1]
        \Require  $a_0,...,a_k$ are vertices of $H$ that lie along one simple path that ends in $a$
        \For{$i=0$ to $k-1$}
        \State Perform a visit from $a_i$ considering $a_{i+1}$ a sink and flag all visited edges
        \If{the visit reaches $a$} \Return {\sc Vuln}
        \EndIf
        \State When no edge can be visited, quit the visit but leave the flag on the edges
        \EndFor
        \State \Return {\sc NotVuln}
    \end{algorithmic}
    \end{algorithm}
    
\begin{proposition}\label{prop:A2andA5andA7}
 Let $H$ be a graph and $a_0,\ldots,a_k$ be $k+1$ nodes of $H$ that lie on a simple path that terminates in node $a$. 
 Then, {\sl NoJumps}$(H,a_0,\ldots,a_k,a)$ returns {\sc Vuln} if and only if
 $\exists a_i\ \exists P=a_i\pat a\ \exists j > i$ such that $a_j \notin P$.
\vspace{-2mm}       
 \end{proposition}
 
\noindent We can now use {\sl NoJumps} with proper inputs to check $A_2$, $A_5$ and $A_7$:
 \begin{itemize}
\vspace{-1mm}       
     \item For $A_2$, we invoke {\sl NoJumps}$(G,\x_1,\omega_1,...,\omega_k,t)$. Indeed, we observe that
 $\exists i,j$ such that $\exists P_i=\x_i\pat t$ in $G\setminus\{\omega_j\}$ if and only if $\exists j$ such that $\exists P=\x_1\pat t$ in $G\setminus\{\omega_j\}$:
the left-to-right implication holds because we can pass on $\C$ from $\x_1$ and reach any $\x_i$;
the other implication holds by definition (choose $i=1$).
     \item For $A_5$, we invoke {\sl NoJumps}$(G \setminus (\hat\xi \overset{\C}\pat \hat\xi),\hat \x,\gamma_1,...,\gamma_k,t)$.
    \vspace*{-1mm}
     \item For $A_7$, we invoke {\sl NoJumps}$(G\setminus(\hat \x\overset{\hat P}{\pat}t],\e^*,\x_1,...,\x_k,\hat \x)$, where $\x_k$ is the exit before $\hat\x$ (at most $k=1$, since we are assuming to have at least two exits).
 \vspace{-3mm}       
\end{itemize}

\begin{algorithm}[t]
    \caption{{\sl OrderedExits}$(G)$}
    \label{alg:orderedexits}
    \begin{algorithmic}[1]
        \Require In $G$ all nodes and edges are labelled as described in Section~\ref{Definitions}
        \State Assume that $|\Xi| = k$
        \For{each node $v$ of $\hat P$}
        \State Define $o[v]$ the position of $v$ in the path
        \EndFor
        \State {\sl ClosGamma} = $o[t']$
        \State {\sl ClosGammaProv} = $o[t']$
        \For{$i=1$ to $k$}
        \State Perform a visit from $\x_i$ by 
        \State\qquad (a) not crossing edges that are flagged, and
        \State\qquad (b) flagging all the newly crossed edges, and
        \State \qquad (c) considering $\hat P$ and $\C \setminus \{\xi_i\}$ as sinks
        
        \ForAll{$v$ touched in the visit}
        \If{$v$ is not marked and $v \notin \hat P \cup \C$} mark $v$ with $i$           
        \EndIf
        \If{$v$ is marked with $h < i$} mark $v$ with $(h,i)$
        \EndIf
        \If{$v\in\hat P$}
        \If{$o[v] > $ {\sl ClosGamma}} \Return {\sc Vuln} /* $A_6(G)${\bf (I)} holds */ \label{line:NotOrdered}
        \EndIf
        \If{$o[v]< $ {\sl ClosGammaProv}} {\sl ClosGammaProv} = $o[v]$
        \EndIf
        \EndIf
        \EndFor
        \State At the end of the visit, do not remove the flags on the crossed edges
        \State {\sl ClosGamma} = {\sl ClosGammaProv}
        \EndFor
    \State Cancel all the flags on edges
    \ForAll{$v \in \hat P$}
    \State Perform a backwards visit from $v$ by
        \State\qquad (a) not crossing edges that are flagged, and
        \State\qquad (b) flagging all the newly crossed edges, and
        \State \qquad (c) considering $\hat P \setminus \{v\}$ and $\C$ as sinks
        \If{find a node marked with a pair} 
        \Return {\sc Vuln} /* $A_6(G)${\bf (II)} holds */ \label{line:NotDisjoint}
        \EndIf
        \State At the end of the visit, do not remove the flags on the crossed edges
        \EndFor
        \State \Return {\sc NotVuln}
    \end{algorithmic}
\end{algorithm}

\vspace{-2mm}       
\subsection{\texorpdfstring{Checking Condition $A_6$}{Checking Condition A6}}
\label{sec:checkingA6}
\vspace{-1mm}

Algorithm~\ref{alg:orderedexits} first assigns to each node in $\hat P$ its order along that path; this is needed to check condition {\bf (I)} of $A_6$ in constant time. 
Then, we start a visit from each of the $k$ exits; in the visit from the generic $\x_i$, we keep track of the visited node of $\hat P$ which is closer to $\hat\x$, if any.
If it exists, this node is one of the $\gamma_i$'s and we record it (actually, its order along $\hat P$) in variable {\sl ClosGamma}; 
in this way, we know that there is a simple path from $\x_i$ to that node, internally disjoint from $\hat P$ (being the nodes of $\hat P$ sinks of the visit).
Since we want that such paths are disjoint also from $\C$, we also consider all the nodes of $\C$, except for $\xi_i$, as sinks.
For every node not belonging to $\hat P \cup \C$ that is touched for the first time in the current visit, we mark it with $i$; for every node that has already been touched in one previous visit (say $h$), we update its mark with the pair $(h,i)$; for every node which has been touched more than once, the mark does not change. Indeed, to check $A_6$ we should be able to ascertain whether there are (at least) two different exits with intersecting simple paths that lead to $\hat P$.

Whenever we reach a node of $\hat P$, we check if that node appears along $\hat P$ after the $\gamma$ with minimum order $o$ calculated in the previous iterations (as recorded in {\sl ClosGamma}); 
if this is the case, then condition {\bf (I)} of $A_6$ is met and the algorithm returns {\sc Vuln}; otherwise, we need to upload the minimum. Since we can reach different nodes of $\hat P$ during the visit from $\xi_i$ and we need to compare their order with the order of the minimum from the previous iterations, we save the (order of the) provisory minimum of the iteration in {\sl ClosGammaProv} and upload {\sl ClosGamma} only at the end of the iteration. 
%
When the visit from $\xi_i$ is completed, since we want a linear-time algorithm, we keep the flagging on the crossed edges and start a new visit from the following exit.

\newpage
After the {\bf for}-loop, the nodes marked with pairs are the candidates to satisfy condition {\bf (II)} of $A_6$; however, not all of them can have paths leading to some $\gamma \in \hat P$. To see this, consider the following example:

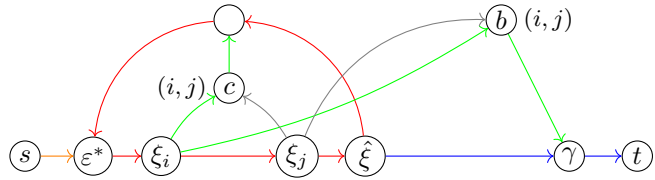
\begin{tikzpicture}[scale=0.90,every node/.style={circle,draw, minimum size=4mm, inner sep=1pt}]
\node (s) at (0,0) {$s$};
\node (es) at (1,0) {$\e^*$};
\node (xi) at (2,0) {$\x_i$};
\node (xj) at (4,0) {$\x_j$};
\node [label=left:{\small $(i,j)$}] (c) at (3,1) {$c$};
\node (up) at (3,2) {};
\node [label=right:{\small $(i,j)$}] (b) at (7,2) {$b$};
\node (xt) at (5,0) {$\hat \x$};
\node (g) at (8,0) {$\gamma$};
\node (t) at (9,0) {$t$};

\draw [->, orange] (s) -- (es);
\draw[->,red] (es) to (xi);
\draw[->, red] (xi) to (xj);
\draw[->, red] (xj) to (xt);
\draw[->,blue] (xt) -- (g);
\draw [->,red,bend right =40] (xt) to (up);
\draw [->,red,bend right =40] (up) to (es);
\draw[->,blue] (g) -- (t);
\draw [->,green, bend left=15] (xi) to (c);
\draw [->,green] (c) to (up);
\draw [->,green, bend right=10] (xi) to (b);
\draw [->,green] (b) to (g);
\draw [->,gray,bend right =15] (xj) to (c);
\draw [->,gray,bend left =35] (xj) to (b);
\end{tikzpicture}

\noindent
In the visit from $\x_i$, the algorithm follows the green edges and marks the nodes not belonging to $\hat P \cup \C$ with $i$; in particular, nodes $b$ and $c$ are marked with $i$. Then, in the visit from $\x_j$, it crosses the grey edges and marks with $j$; in particular, $b$ and $c$ are now marked with the pair $(i,j)$. 
However, $b$ is a point where two simple paths $\x_i\pat\gamma$ and $\x_j\pat\gamma$ (both internally disjoint from $\hat P$ and $\C$) intersect, whereas $c$ is not.
Hence, the aim of the final backward visit performed by the algorithm starting from all the nodes of $\hat P$ is exactly to identify any node marked with a pair that has a path to some $\gamma \in \hat P$, if any.
In the previous example, the algorithm finds, during the backward visit from $\gamma$, only $b$, and so it correctly returns {\sc Vuln} in line \ref{line:NotDisjoint}; in contrast it (correctly) does not find $c$.
Also in this backward visit, we avoid crossing the same edge more than once, by keeping the flag on the edges already crossed.
This ensures that the overall cost of the algorithm is $O(m)$, as it consists in two visits of the graph.
\vspace{-1mm}

\begin{proposition}\label{prop:A6}
 {\sl OrderedExits}$(G)$ returns {\sc Vuln} if and only if $A_6(G)$.
 \vspace{-2mm}       
\end{proposition}


 \begin{algorithm}[t]
    \caption{{\sl FindACycle}$(G)$}
    \begin{algorithmic}[1]
        \State Find all the strong connected components ($SCCs$) of $G$
        \If{all $SCCs$ are trivial (the graph is acyclic)}
        \Return $acyclic$
        \EndIf
        \State Perform a $\mathit{BFS}$ from $s$ to find a closest vertex 
        to $s$ that belongs to a non-trivial $SCC$; label it as $\e^*$
        \State Label all the nodes and edges of the shortest path from $s$ to $\e^*$ with $P^*$
        \State Find a cycle containing $\e^*$ and label all its nodes and edges with $\C$
        \State\Return $cyclic$
    \end{algorithmic}
    \label{alg:FindtheCycle}
\end{algorithm}

\section{A Linear-time Amortized Decremental Algorithm}
\label{sec:amortized}
\vspace{-2mm}       

To get a linear-time amortized algorithm, we 
provide suitable versions of functions {\sl FindACycle} and {\sl Search$\mathcal{W}$}, called in Algorithm~\ref{alg:cgs} of Section~\ref{sec:decAlg}. 
As in \cite{CGS19}, it is essential that the cycle choosen by {\sl FindACycle} is (one of) the cycle(s) closest to $s$: 
this is what Algorithm \ref{alg:FindtheCycle} does, by following the  decomposition of a digraph in its strongly connected components (SCCs).
It works in $\mathcal{O}(m)$  
since finding SCC is linear \cite{RT71} and the other operations are just visits. 

Once we found $\C$ (and, correspondingly, $\e^*$ and $P^*$), we invoke the functions developed in Section~\ref{sec:algos} to label the graph and analyze the cycle; this is done by function {\sl Search$\mathcal{W}$} in Algorithm \ref{alg:CycleAnalysis}, by using Algorithms \ref{alg:Findthenodes}, \ref{alg:nobadchord}, \ref{alg:nojumps} and \ref{alg:orderedexits}. Since all the algorithms invoked cost $\mathcal{O}(m)$, the cost of the algorithm itself is $\mathcal{O}(m)$.

\begin{theorem}
If Algorithm \ref{alg:CycleAnalysis} returns {\sc Vuln}, then the net is vulnerable.
If Algorithm \ref{alg:CycleAnalysis} returns $\hat e$, then $\hat e$ is redundant.
\vspace{-1mm}
\end{theorem}
\begin{proof}
If the algorithm returns {\sc Vuln}, then one among  Algorithms~\ref{alg:Findthenodes}, \ref{alg:nobadchord}, \ref{alg:nojumps} or \ref{alg:orderedexits} returned {\sc Vuln}. By Propositions~\ref{prop:A1andA4}, \ref{prop:A3}, \ref{prop:A2andA5andA7} and \ref{prop:A6}, there exists some $i \in \{1,\ldots,7\}$ s.t. $A_i(G)$ holds; thus, $G$ is vulnerable, by one among Propositions~\ref{prop:noxi}, \ref{prop:onlydirectentry}(1), \ref{prop:nobadchord}, \ref{prop:noyi}, \ref{prop:jumponpt}, \ref{prop:orderedexits} and \ref{prop:noexitjumponC}.
Conversely, 
if the algorithm returns $\hat e$, then either there is just one exit, and so by Remark~\ref{rem:oneExit} $\hat e$ is redundant, or none of Algorithms~\ref{alg:Findthenodes}, \ref{alg:nobadchord}, \ref{alg:nojumps} or \ref{alg:orderedexits} returned {\sc Vuln}. By Propositions~\ref{prop:A1andA4}, \ref{prop:A3}, \ref{prop:A2andA5andA7} and \ref{prop:A6}, we know that $A_i(G)$ does not hold, for any $i \in \{1,\ldots,7\}$. By Theorem \ref{th:Main}, $\hat e$ is redundant.
\vspace{-1mm}
\end{proof}

\begin{algorithm}[t]
    \caption{{\sl Search}$\mathcal{W}
    (G)$}
    \begin{algorithmic}[1]
        \Require $\e^*$ and all the nodes and edges of $\C$ and $P^*$ are labelled in $G$
        \If{{\sl FindNodes}$(G)$ returns {\sc Vuln}} \Return {\sc Vuln}
        \EndIf
        \If {it returns $\hat e$}
        \Return {$\hat e$}
        \EndIf
        \State \textbf{if} any of
         \State \qquad {\sl NoBadChord}$(G)$ 
         \State \qquad {\sl NoJumps}$(G,\x_1,\omega_1,...,\omega_k,t)$
         \State \qquad {\sl NoJumps}$(G\setminus (\hat\xi \overset{\C}\pat \hat\xi),\hat \x,\gamma_1,...,\gamma_k,t)$
         \State \qquad {\sl NoJumps}$(G\setminus(\hat \x\overset{\hat P}{\pat}t],\e^*,\x_1,...,\x_k,\hat \x)$
         \State \qquad {\sl OrderedExits}$(G)$
         \State \quad returns {\sc Vuln} 
         \State \textbf{then} \Return {\sc Vuln}
        \State \textbf{else} \Return $\hat e$
            
    \end{algorithmic}
    \label{alg:CycleAnalysis}
\end{algorithm}

As described in Section~\ref{sec:decAlg}, function {\sl isVuln} in 
Algorithm~\ref{alg:cgs} is a decremental algorithm, just calling it on the graph $G\setminus \{e\}$, being $e$ the edge to be removed, while 
Algorithm~\ref{alg:dec} 
summarizes how to keep information about redundant edges already discovered by function {\sl isVuln} in a sequence of edge removals.
Function {\em isVuln} is invoked $k\in\mathcal{O}(m)$ times; 
the total cost of such 
invocations is $\mathcal{O}(m^2)$, because it consists in at most $\mathcal{O}(m)$  invocations of Algorithms~\ref{alg:FindtheCycle} and~\ref{alg:CycleAnalysis}. 
Hence, the decremental test for vulnerability has an amortized complexity of $\mathcal{O}(m)$ for every removed edge, by assuming $k\in\mathcal{O}(m)$ edge removals.
\vspace{-2mm}

\section{Conclusions}
\label{sec:concl}
\vspace{-2mm}       

In this paper, we have improved the ‘‘hunting $\mathcal{W}$ in cycles'' strategy underlying the algorithm in \cite{CGS19} and shown that it can work in $\mathcal{O}(m^2)$ (instead of $\mathcal{O}(nm^2)$). Moreover, we have shown that, irrespectively of which of the two versions one wants to use, the approach can be used to develop a decremental algorithm, able to ascertain the vulnerability condition upon edges removals. 

This paper complements a companion paper \cite{fgs26a}, 
where an incremental algorithm is presented; these can be used when the considered net evolves in time by removing/adding connections among nodes.
Interestingly, the incremental algorithm in \cite{fgs26a} 
and the decremental algorithm presented here are in some sense dual:
the incremental one is based on the series-parallel irredundant non-vulnerable part of the graph and takes advantage of the algorithm in \cite{MS23};
the decremental one is based on hunting redundant edges and a cycle analysis along the lines of \cite{CGS19}.
In both cases, the static algorithms are $\mathcal{O}(m^2)$, whereas the dynamic algorithms are $\mathcal{O}(m)$ amortized for each edge insertion/removal.

The obvious direction for future research is the development of a fully dynamic approach, where edge additions and removals may interleave in an arbitrary way. As usual in the setting of dynamic algorithms, this is the most challenging problem. In our case, vulnerability is preserved under edge additions and non-vulnerability is preserved under edge removals; such a ‘monotonicity' is lost if we interleave additions and removals.


%
%
%
%

\bibliography{threeForms,chan}


\appendix

\section{Traffic Networks, Braess Paradox and Vulnerability}
\label{app:braess}

The set of $st$-paths in $G$ is denoted by $\PP(G)$; 
the set of 
simple $st$-paths is denoted by $\SP(G)$.

A \emph{flow} for a net $(G,s,t)$ is a function 
$\varphi:\SP(G) \rightarrow \mathbb{R}^+$.
The value of a flow is the sum of the values sent over all paths. 
A flow induces a unique flow $\varphi(e)$ on edges: for any edge $e\in E$, 
$\varphi(e)=\sum_{P\in \SP(G): e\in P}\varphi(P)$. 

A {\em latency function} $l_e:\mathbb{R}^+\rightarrow \mathbb{R}^+$ 
assigns to each edge $e$ a latency as a function of the flow in it;
as usual, we only consider continuous and non-decreasing latency functions.

The latency of an $st$-path $P$ under a flow $\varphi$ is the sum of the latencies of 
all edges in the path under $\varphi$, i.e., $l_P(\varphi)=\sum_{e\in P} l_e(\varphi(e))$.
If $H$ is a subgraph of $G$, we denote with 
$l|_H$ the restriction of the latency function $l$ on the edges of $H$.

Given a net $G$, a real number $r\in\mathbb{R}^+$ and a latency function $l$, 
we call the triple $(G, r, l)$ an {\em instance} (of $G$). A flow $\varphi$ is {\em feasible} for 
$(G, r, l)$ if the value of $\varphi$ is $r$. Notice that, since we do not have any
constraint on edges or vertices, every $r$ admits at least one feasible flow.

A feasible flow $\varphi$ for $(G, r, l)$ is at {\em Wardrop equilibrium} 
(or is a {\em Wardrop flow}) if, for all $st$-paths $P$ and $Q$  
such that $\varphi(P)>0$, we have $l_P(\varphi)\leq l_Q(\varphi)$.
In particular, this implies that, if $\varphi$ is a Wardrop flow, 
all $st$-paths to which $\varphi$ assigns a positive flow
have the same latency. It is known \cite{Rough06} that every
instance admits a Wardrop flow and that different Wardrop flows for the same
instance have the same latency along all $st$-paths with a positive flow. Thus,
we denote with $L(G, r, l)$ the latency of all $st$-paths 
with positive flow at Wardrop equilibrium; whenever $r = 0$,
we let $L(G, r, l)$ be 0.

{\em Braess paradox} \cite{braessFormal,braessOriginal} originates when latency 
at Wardrop equilibrium decreases because of removing some edges (or equivalently, by increasing 
latency functions on them): an instance $(G, r, l)$ suffers from Braess paradox 
if a subgraph of $G$ exists with a lower latency.
In \cite{Rough06} it is shown that 
it is \NP-hard to decide whether or not an instance of a net  
suffers from the Braess paradox. A question raised in \cite{Rough06} is to study
the problem from a graph-theoretical perspective, 
that is whether a net admits instances suffering from 
the paradox.
This leads to the following definition \cite{Rough06}.
\begin{definition}
	A net $G$ is {\em vulnerable} if there exist a value $r$, a latency function
	$l$ and a subgraph $H$ of $G$ such that $L(G, r, l)>L(H, r, l|_H)$.
\end{definition}

To conclude, we also report, for the sake of completeness, the definition of (two-terminal) series-parallel graphs.

\begin{definition}[Two-Terminal Series-Parallel Graphs (TTSP) \cite{Duf65}]
A graph $G=(V,E)$, with two distinguished vertices $s$ and $t$, is a {\em TTSP-graph} if it can be turned into the single edge $(s,t)$ by a sequence of the following operations:
\begin{itemize}
\item Replace a pair of parallel edges with a single edge that connects their common endpoints;
\item Replace a pair of edges $(u,w)$ and $(w,v)$ with the edge $(u,v)$, whenever $w$ has degree 2 and it is different from $s$ and $t$.
\end{itemize}
\end{definition}

\section{Omitted Proofs from Section~\ref{sec:algos}}
\label{app:proofs}

\begin{proof}[Proof of Proposition~\ref{prop:A1andA4}]
It suffices to observe that $A_1(G)$ means that there is a path $P=s\pat x$, with $x\in(\x_1\overset{\C}{\pat}\hat{\x}\overset{\hat P}{\pat}t')$ and $\mathring P\cap (\C\cup\hat P)=\emptyset$; this holds if and only if the visit from $s$ on line 6 reaches a node which belongs to $(\x_1\overset{\C}{\pat}\hat{\x}\overset{\hat P}{\pat}t')$. 
Similarly, $A_4(G)$ means that there is a path $P=\x_1\pat y$, with $y\in(t'\overset{\hat P}{\pat}t]$ and $\mathring P\cap \hat P=\emptyset$; this holds if and only if the visit from $\x_1$ on line 10 reaches a node which belongs to $(t'\overset{\hat P}{\pat}t]$. 
\end{proof}

\begin{proof}[Proof of Proposition~\ref{prop:A3}]
($\Leftarrow$)
$A_3(G)$ means that $\e_1 \in (\hat\xi \pat \e^*)$ and $\exists a\in(\e^*\overset{\C}{\pat}\hat\x]$ s.t. $\exists P=\e_1\pat a$ in $G\setminus\{\e^*\}$. 
Then, the algorithm reaches a node in $(\e^*\overset{\C}{\pat}\hat\x]$ with a visit in $G\setminus\{\e^*\}$ started from $\e_1$ and returns {\sc Vuln} (line 3).
 
($\Rightarrow$) If the algorithm returns {\sc Vuln}, then $\e_1\not\in [\e^*\overset{\C}{\pat}\hat\x]$ and
the visit is  performed; the path $P$ required by $A_3$ is then composed by all the nodes and edges found backtracking from the node in $(\e^*\overset{\C}{\pat}\hat\x]$ to $\e_1$.
\end{proof}

 \begin{proof}[Proof of Proposition~\ref{prop:A2andA5andA7}]
     $(\Rightarrow)$ If the algorithm answers {\sc Vuln}, then the visit starting from some $a_i$ reached $a$ without touching $a_{i+1}$.
     
     $(\Leftarrow)$ We prove the contrapositive. During the $i^{th}$ iteration of the {\bf for}-loop, the following invariant holds: 
\textit{Each simple path that starts from a visited node and that in $H$ can reach $a$ passes in $a_{i+1}$}. 
This is true since the only possible sink of iteration $i$ is $a_{i+1}$ (since we are assuming that the algorithm returns {\sc NotVuln}) and the iteration $i+1$ starts from $a_{i+1}$ (the previous sink), continuing all the branches of the visit previously stopped there. 
So, at the $k^{th}$ iteration, all the simple paths starting from each $a_i$ that could reach $a$ in $H$ pass in $a_k$. 
This means that all the simple paths starting from any $a_i$ touches all the $a_j$ with $j>i$.
\end{proof}

\begin{proof}[Proof of Proposition~\ref{prop:A6}]
    $(\Rightarrow)$
    If the algorithm returns {\sc Vuln} in line~\ref{line:NotOrdered}, then it has found two disjoint paths, from two different exits (say, $\xi_i$ and $\xi_j$, with $i<j$), that land to different nodes of $\hat P$ (say, $\gamma$ and $\gamma'$, respectively) such that $o[\gamma] < o[\gamma']$; this means that $\gamma <_{\hat P} \gamma'$ and so condition {\bf (I)} of $A_6$ holds. Notice that the algorithm updates $ClosGamma$ only at the end of each iteration; so, it is not possible that it answers {\sc Vuln} for two paths from the same $\x_i$.

Otherwise, the algorithm must return {\sc Vuln} in line 25; then, there exists some $b$ marked with a pair $(i,j)$ (with $i < j$) backward reachable from some $\gamma \in \hat P$. Hence, there exist two paths $P_i'=\x_i\pat b$  
and $P_j'=\x_j\pat b$ such that all internal nodes of $P_j'$ are marked with $j$ but not with $i$; so, $P_i'$ and $P_j'$ are internally disjoint. 
Since $b$ is marked with $i$, there exists a $P_{\gamma}=b\pat\gamma$ with all nodes marked with a pair that has at least one marker that is $i$ or less. 
Then, the path $P'=\x_j\overset{P_j}{\pat}b\overset{P_{\gamma}}{\pat}\gamma$ is simple, because of the structure of the markers. 
Let $a$ be the first node of $P_{\gamma}$ that is in a simple path $P''=\x_i\pat a\pat \gamma$ (at least one exists, possibly $a=b$ --- see the graph depicted in Section~\ref{sec:checkingA6}). 
Then, the two paths $P_i=\x_i\overset{P''}{\pat}a\overset{P_{\gamma}}{\pat}\gamma$ and $P_j=\x_j\overset{P_j}{\pat}b\overset{P_{\gamma}}{\pat}\gamma$ are both simple and internally intersect in $a$; thus, condition {\bf (II)} of $A_6$ holds.\footnote{
A minimal example where $a \neq b$ is depicted below:

\begin{tikzpicture}[scale=0.95,every node/.style={circle,draw, minimum size=4mm, inner sep=1pt}]
\node (s) at (0,0) {$s$};
\node (es) at (1,0) {$\e^*$};
\node (xi) at (2,0) {$\x_i$};
\node (xj) at (4,0) {$\x_j$};
\node [label=right:{\small \color{pink} $(i,j)$}] (b) at (7,2) {$b$};
\node [label=right:{\small \color{pink} $i$}](a) at (7.5,1.3) {$a$};
\node (xt) at (5,0) {$\hat \x$};
\node (g) at (8,0) {$\gamma$};
\node (t) at (9,0) {$t$};

\draw [->, orange] (s) -- (es);
\draw[->,red] (es) to (xi);
\draw[->, red] (xi) to (xj);
\draw[->, red] (xj) to (xt);
\draw[->,blue] (xt) -- (g);
\draw [->,red,bend right =40] (xt) to (es);
\draw[->,blue] (g) -- (t);
\draw [->,green, bend left=10] (xi) to (a);
\draw [->,green, bend right =20] (a) to (b);
\draw [->,green, bend right =20] (b) to (a);
\draw [->,green] (a) to (g);
\draw [->,gray,bend left =35] (xj) to (b);
\end{tikzpicture}

%

\noindent
The internally intersecting paths are $P_i = \xi_i \pat a \pat \gamma$ and $P_j = \xi_j \pat b \pat a \pat \gamma$.
}
    
    $(\Leftarrow)$
First, assume that $A_6(G)$ holds because of condition {\bf (I)}. If there exist two different exits $\x_i$ and $\x_j$ with $i<j$ and two disjoint paths $\x_i\pat \gamma$ and $\x_j\pat \gamma'$ with $\gamma <_{\hat P} \gamma'$, then the algorithm first found $\gamma$ and recorded in $ClosGamma$ its order in $\hat P$ (or a smaller order, if the visit from $\xi_i$ found another node of $\hat P$ preceding $\gamma$) at the end of iteration $i$ of the {\bf for}-loop. In the following iterations, $ClosGamma$ cannot increase; so, when the algorithm finds $\gamma'$ in iteration $j$, it has that $o[\gamma'] > o[\gamma] \ge ClosGamma$ and returns {\sc Vuln} in line~\ref{line:NotOrdered}.

Finally, assume that condition {\bf (II)} of $A_6(G)$ holds; then, there are two paths $P_i:\x_i\pat\gamma$ and $P_j=\x_j\pat\gamma'$ (for $i<j$ and $\gamma,\gamma' \in \hat P$) that have their first intersection in a node $a$. 
    If $a$ is marked with a pair, then, since $\gamma$ is reachable from $a$, the backward visit from $\gamma$ will find $a$ and the algorithm returns {\sc Vuln} in line~\ref{line:NotDisjoint}. Assume that $a$ is not marked with a pair. Since $a$ is reachable from $\x_j$, if $a$ has not been marked with $j$, then the visit from $\x_j$ stopped in a node $b$ preceding $a$ and marked with a pair; since $b$ is reachable from $a$, the backward visit from $\gamma$ reaches $b$ and the algorithm returns {\sc Vuln}, like above.\footnote{By considering the example in the previous footnote, $a$ is not marked with $(i,j)$ because, when the visit from $\xi_j$ arrives to $b$ (and marks it with $(i,j)$), all edges outgoing from $b$ are already flagged and so the visit from $\xi_j$ stops in $b$ without reaching $a$.
    } 
If $a$ is marked with $j$ but not with $i$, the proof is analogous. 
\end{proof}

\end{document}